\documentclass{article}

\usepackage[T1]{fontenc}
\usepackage{authblk}
\usepackage{cite}
\usepackage[hidelinks]{hyperref}

\usepackage{amsmath,amsthm}
\usepackage{amssymb}
\usepackage{physics}

\usepackage{fancyhdr}
\usepackage[top=1.5cm, bottom=1.5cm, left=1.5cm, right=1.5cm]{geometry}
\usepackage{xcolor}
\usepackage{graphicx}
\usepackage{tikz}
\usetikzlibrary{shapes, arrows, positioning}

\newtheorem{definition}{Definition}
\newtheorem{proposition}[definition]{Proposition}
\newtheorem{lemma}[definition]{Lemma}

\newtheorem{theorem}[definition]{Theorem}
\newtheorem{corollary}[definition]{Corollary}

\newtheorem{remark}[definition]{Remark}
\newtheorem{example}[definition]{Example}

\def\r{\rho}

\newcommand{\nc}{\newcommand}

\definecolor{strongred}{RGB}{255, 0, 0}
\definecolor{stronggreen}{RGB}{0, 155, 0}
\definecolor{strongblue}{RGB}{0, 0, 255}
\nc{\red}[1]{\textcolor{strongred}{#1}}
\nc{\green}[1]{\textcolor{stronggreen}{#1}}
\nc{\blue}[1]{\textcolor{strongblue}{#1}}

\nc{\bbC}{\mathbb{C}}
\nc{\bbH}{\mathbb{H}}
\nc{\bbM}{\mathbb{M}}
\nc{\bbR}{\mathbb{R}}
\nc{\tbc}{\blue{\textbf{to be continued ...}}}

\begin{document}

\title{Construction of three-qubit positive-partial-transpose entangled states of rank four}
\author{Yonggang Cheng, Lin Chen\footnote{Corresponding author: linchen@buaa.edu.cn}}
\affil{LMIB (\href{https://ror.org/00wk2mp56}{Beihang University}), Ministry of Education, and School of Mathematical Sciences, Beihang University, Beijing 100191, China}
\date{\today}

\maketitle

\begin{abstract}
\noindent Multiqubit positive-partial-transpose (PPT) entangled states provide canonical examples of bound entanglement and stringent test cases for multipartite separability criteria. We study PPT entangled states of the minimum possible rank in the three-qubit system, namely rank four. In the established terminology, states with a nonzero Lorentz invariant are called type I, whereas those with a zero Lorentz invariant are called type II. Every such state constructed from an unextendible product basis (UPB) has a nonzero Lorentz invariant and hence is of type I. We give a necessary and sufficient criterion for determining whether a type-I state is SLOCC-equivalent to a PPT entangled state constructed from a UPB. For type-II states, we obtain a Bell-basis normal form and prove that it exhausts all such states up to SLOCC equivalence. We further give a necessary and sufficient condition for two states in this normal form to be SLOCC-equivalent. We also determine the range of the Lorentz invariant for multiqubit states.

\vspace{0.5cm}

\end{abstract}


\vspace{1cm}

\section{Introduction}

One of the main problems in quantum information theory is to determine the separability of a quantum state. It can be challenging to detect a multipartite entangled state, which is a valuable resource for quantum technologies \cite{length}. Peres \cite{Peres.A} proposes that a necessary condition for the separability is positive partial transpose (PPT), which is proven to be also a sufficient condition for $2\times2$ and $2\times3$ systems \cite{Horodecki.MPR}. Throughout this paper, an $n$-partite state $\r$ is called PPT if
\begin{equation}
    \r^{T_S}\geq0
\end{equation}
for every bipartition $S|\bar{S}$, where $T_S$ denotes the partial transpose with respect to the subsystems in $S$. In particular, if $A$, $B$ and $C$ label the three qubits, the three bipartitions are $A|BC$, $B|AC$ and $C|AB$. Hence, a three-qubit state is PPT if
\begin{equation}
    \r^{T_A}\geq0,\quad\r^{T_B}\geq0,\quad\r^{T_C}\geq0.
\end{equation}
Ref. \cite{construction} provides a sufficient condition for certifying that a state is PPT in terms of the second-order moment of its partial transpose. There exist PPT states that are not separable in more general systems, namely PPT entangled states (PPTESs). The first examples of PPTES in $3\times3$ and $2\times4$ systems are provided in \cite{Horodecki.P}. Moreover, Ref. \cite{multiqubit} provides two methods for generating multiqubit PPTES and Ref. \cite{multiqutrit} proves the existence of PPTES in three or more qutrits by explicit construction. Multipartite bound entanglement that is PPT with respect to every bipartition has also been identified in thermal spin models, including three-qubit systems, using spin-squeezing inequalities \cite{Toth2007}. Some PPTESs can provide a significant advantage over separable states in quantum metrology \cite{powerful}. In three-qubit system, any PPT state of rank at most three is known to be separable \cite{low_rank}, i.e., rank four is the minimal rank for the existence of a three-qubit PPTES. A notable property of three-qubit rank-four PPTES is that they are separable with respect to every bipartition, although they are not fully separable. Moreover, for every bipartition, their decompositions into pure bipartite product states are unique \cite{biseparable,Chen_2013}. Consequently, these states are not genuinely tripartite entangled. Each partial transpose of such states still has rank four and their ranges have no tripartite product vector \cite{444,Chen_2013,prod}. Furthermore, every three-qubit rank-four PPTES is an extreme point of the convex set of PPT states \cite{Lorentz_invariant}. A three-nuclear-spin NMR experiment prepared and detected pseudo-bound-entangled states for a fixed qubit–ququart bipartition using measurements of three Pauli observables and quantum-state tomography \cite{experiment1}. Although the target state is not PPT across all three bipartitions, the experiment provides a possible framework for preparing and verifying the states studied here. Nevertheless, a complete construction of all three-qubit rank-four PPTESs remains unknown.

An unextendible product basis (UPB) is an orthonormal set of product vectors spanning a proper subspace whose orthogonal complement contains no product vector \cite{Bennett}. An orthonormal product basis that spans the whole space is regarded as a trivial UPB. If $\{\xi_i\}_{i=1}^m$ is a nontrivial UPB in a $d$-dimensional multipartite Hilbert space, then the normalized projector onto its orthogonal complement
\begin{equation}
    \r_{UPB}=\frac{1}{d-m}\bigg(I_d-\sum_{i=1}^m\ket{\xi_i}\bra{\xi_i}\bigg)
\end{equation}
is PPT with respect to every bipartition and is entangled. Hence, UPBs provide a systematic method for constructing PPTESs. For example, a complete construction of $3\times3$ UPB is achieved \cite{3x3upb} and this tool helps construct all rank-four PPTES in $3\times3$ system \cite{3x3classify,Lin3x3}. Refs. \cite{concurrence,max} further investigate the concurrence of PPTESs constructed from UPBs and UPBs of maximal size, respectively. The structure of multiqubit UPBs becomes considerably more complicated as the number of qubits increases \cite{Johnston}. In the three-qubit case, every nontrivial UPB contains exactly four product vectors and, up to local unitaries and a permutation of its members, admits the three-parameter canonical form presented in \eqref{eq:three-qubit UPB} \cite{three_qubit_UPB}. However, there exist rank-four PPTESs that are not constructed from UPBs \cite{general_position}. Therefore, the structure of such states is not simple.

A general stochastic LOCC protocol may involve noninvertible local filters and can therefore map an entangled state to a separable state. In this work, two three-qubit states $\r_1$ and $\r_2$ are said to be SLOCC-equivalent if there exist invertible local operators $A,B,C\in GL(2,\bbC)$ such that
\begin{equation}
    \r_2=\frac{G\r_1G^\dagger}{\Tr(G\r_1G^\dagger)},\quad G=A\otimes B\otimes C.
\end{equation}
Since the inverse transformation also consists of invertible local operators, SLOCC equivalence preserves both entanglement and the PPT property. Under this equivalence, pure three-qubit states fall into six classes: one fully separable class, three biseparable classes corresponding to the three bipartitions, and two genuinely tripartite-entangled classes represented by the W and GHZ states \cite{WandGHZ}. These classes and their properties have been further investigated in Refs. \cite{SLOCC,GHZ}.

To organize three-qubit rank-four PPTESs under SLOCC equivalence, Ref. \cite{Lorentz_invariant} used a quadratic Lorentz invariant. For three qubits, the invariant $I_\r$ adopted in this work is proportional to the quadratic invariant $I_2(\r)$ in Ref. \cite{Lorentz_invariant}, namely $I_2(\r)=-\frac{1}{8}I_\r$. Although $I_\r$ is strictly invariant only under local $SL(2,\bbC)$ transformations, it is covariant under local $GL(2,\bbC)$ transformations. Consequently, its vanishing or nonvanishing is preserved under SLOCC equivalence. Ref. \cite{Lorentz_invariant} calls a three-qubit rank-four PPTES type I when $I_\r\neq0$ and type II when $I_\r=0$. The formal definition and transformation law of $I_\r$ are given in Sec. \ref{sec:Lorentz invariant}. Other Lorentz invariants for three-qubit pure states, distinct from the quadratic invariant used here, are investigated in Ref. \cite{pure}.

Ref. \cite{Lorentz_invariant} presents a parametrized family that reproduces all type-II states found numerically, but does not prove that every type-II state is locally equivalent to a member of that family. In this work, we prove that every PPTES constructed from a three-qubit UPB is of type I and give a necessary and sufficient criterion for determining whether a type-I state is SLOCC-equivalent to a UPB-derived PPTES. For type-II states, we introduce a Bell-basis normal form, prove that it exhausts all such states up to SLOCC equivalence, and determine the exact SLOCC equivalence relation among its parameters. We also explicitly connect this normal form with the family in Ref. \cite{Lorentz_invariant} and determine the range of the Lorentz invariant for multiqubit states.

The remainder of this paper is organized as follows. Sec. \ref{sec:pre} develops the linear-algebraic and quantum-information tools used in our analysis, including general position, criteria for SLOCC equivalence, and the Lorentz invariant, and presents three examples used in the subsequent analysis. Sec. \ref{sec:res} presents our main results: a necessary and sufficient criterion for a type-I state to be SLOCC-equivalent to a state constructed from a UPB, an exhaustive Bell-basis normal form and the exact SLOCC equivalence classification of type-II states, and the complete range of the Lorentz invariant over all $n$-qubit states. Sec. \ref{sec:con} concludes the paper.

\section{Preliminaries}
\label{sec:pre}

In this section, we introduce the preliminary material used in this paper. It is divided into four subsections: linear algebra in Sec. \ref{sec:linear algebra}, quantum information in Sec. \ref{sec:quant information}, Lorentz invariant in Sec. \ref{sec:Lorentz invariant} and example in Sec. \ref{sec:example}.

\subsection{Linear algebra}
\label{sec:linear algebra}

A three-qubit system can be viewed as bipartite with respect to any of the cuts $A|BC$, $B|AC$ and $C|AB$. To represent product operators across the $B|AC$ cut while retaining the standard $A\otimes B\otimes C$ ordering, we introduce the following operation.
\begin{definition}
\label{de:split tensor product}
Let
\begin{eqnarray}
P\in\bbM_2(\bbC),\quad Q=\begin{pmatrix}
    A & B \\
    C & D
\end{pmatrix}\in\bbM_4(\bbC),
\end{eqnarray}
where $A,B,C,D\in\bbM_2(\bbC)$. Let
\begin{eqnarray}
P\otimes_sQ:=\begin{pmatrix}
    P\otimes A & P\otimes B \\
    P\otimes C & P\otimes D
\end{pmatrix}.
\end{eqnarray}
Then $P\otimes_sQ=V\otimes P\otimes W$ if $Q=V\otimes W$, where $V,W\in\bbM_2(\bbC)$. We call $\otimes_s$ the split tensor product.
\end{definition}

\begin{definition}
\label{de:GP}
A finite set
\begin{eqnarray}
\{z_i=x_{1i}\otimes x_{2i}\otimes\dots\otimes x_{ni}\in\bbC^{d_1}\otimes\bbC^{d_2}\otimes\dots\otimes\bbC^{d_n}:i=1,2,\dots,k\}
\end{eqnarray}
is said to be in general position if for each $j=1,2,\dots,n$ and every subset $J\subset\{1,2,\dots,k\}$ with $|J|\leq d_j$, the set $\{x_{ji}:i\in J\}$ is linearly independent in $\bbC^{d_j}$. 
\end{definition}

By Theorem 3.2 in \cite{general_position}, we have the following claim.
\begin{theorem}
\label{th:GP}
Let $\{z_i:i=1,2,\dots,k\}$ be a set of product vectors in general position in $\bbC^{d_1}\otimes\bbC^{d_2}\otimes\dots\otimes\bbC^{d_n}$. If
\begin{equation}
k\leq\sum_{i=1}^n(d_i-1),
\end{equation}
the span of these product vectors has no more product vectors except for scalar multiples of these product vectors.
\end{theorem}

From the following lemma, we can see that four pairwise linearly independent two-dimensional vectors are determined by a single complex parameter up to a common invertible linear transformation, individual nonzero rescalings, and permutations.
\begin{lemma}
\label{le:four vectors}
Let
\begin{eqnarray}
x=(x_1,x_2,x_3,x_4)=\begin{pmatrix}
    a & c & e & g \\
    b & d & f & h
\end{pmatrix}
\end{eqnarray}
with $\det(x_i,x_j)\neq0$ for $i\neq j$.

(i) There are invertible matrices $W\in\bbM_2(\bbC)$ and invertible diagonal $D\in\bbM_4(\bbC)$, such that
\begin{eqnarray}
\label{eq:standard form}
W(x_1,x_2,x_3,x_4)D=\begin{pmatrix}
    1 & 0 & 1 & t \\
    0 & 1 & 1 & 1
\end{pmatrix},\ t=\frac{\det(x_1,x_3)\det(x_2,x_4)}{\det(x_1,x_4)\det(x_2,x_3)}\neq0,1.
\end{eqnarray}

(ii) Let $\sigma$ be a permutation of $\{1,2,3,4\}$ and by (i), there are $W_\sigma$ and $D_\sigma$ such that
\begin{eqnarray}
\label{eq:permutation and transformation}
W_\sigma(x_{\sigma(1)},x_{\sigma(2)},x_{\sigma(3)},x_{\sigma(4)})D_\sigma=\begin{pmatrix}
    1 & 0 & 1 & t_\sigma \\
    0 & 1 & 1 & 1
\end{pmatrix}.
\end{eqnarray}
Then the set of $t_\sigma$ for all $\sigma$ is
\begin{eqnarray}
\label{eq:set of t}
\{t,\frac{1}{t},1-t,\frac{1}{1-t},1-\frac{1}{t},\frac{t}{t-1}\}.
\end{eqnarray}

(iii) Let the sets
\begin{eqnarray}
\label{eq:A,B,C}
\mathcal{A}=(-\infty,0),\quad\mathcal{B}=(0,1),\quad\mathcal{C}=(1,\infty).
\end{eqnarray}
Then there is an invertible transformation $V$ such that $x_i$'s are transformed into two sets of orthogonal vectors if and only if $t\in\mathcal{A}\cup\mathcal{B}\cup\mathcal{C}$. In this case, we have
\begin{equation}
\label{eq:t in ABC}
\begin{aligned}
t\in\mathcal{A},\quad &then\ Vx_1\perp Vx_2,\ Vx_3\perp Vx_4; \\
t\in\mathcal{B},\quad &then\ Vx_1\perp Vx_4,\ Vx_2\perp Vx_3; \\
t\in\mathcal{C},\quad &then\ Vx_1\perp Vx_3,\ Vx_2\perp Vx_4.
\end{aligned}
\end{equation}
\end{lemma}

The proof of the lemma above is given in Appendix \ref{app:proof of four vectors}. The following fact is known in \cite{Bell_states}.
\begin{lemma}
\label{le:permutations of Bell states}
Let $\phi_j$'s be the following four unnormalized Bell vectors
\begin{eqnarray}
(\phi_1,\phi_2,\phi_3,\phi_4)=\begin{pmatrix}
    1 & 1 & 0 & 0 \\
    0 & 0 & 1 & 1 \\
    0 & 0 & 1 & -1 \\
    1 & -1 & 0 & 0
\end{pmatrix}.
\end{eqnarray}
For any permutation $\sigma$ of $\{1,2,3,4\}$, there exists a unitary operator of the product form
\begin{eqnarray}
U=P\otimes Q,\quad P,Q\in U(2)
\end{eqnarray}
and phases $\beta_j\in[0,2\pi)$ such that
\begin{eqnarray}
U\ket{\phi_j}=e^{i\beta_j}\ket{\phi_{\sigma(j)}},\quad j=1,2,3,4.
\end{eqnarray}
\end{lemma}

\subsection{Quantum information}
\label{sec:quant information}

\begin{definition}
A quantum state $\rho$ on the Hilbert space $\mathcal{H}=\bbC^{d_1}\otimes\bbC^{d_2}\otimes\dots\otimes\bbC^{d_n}$ is said to be fully separable if it is a convex combination 
\begin{equation}
\label{eq:separable state}    
\rho=\sum_{i \in I}p_i\ket{z_i}\bra{z_i},\quad p_i>0,\quad\sum_{i\in I}p_i=1,
\end{equation}
where, for every $i\in I$,
\begin{equation*}
\ket{z_i}=\ket{x_{1i}}\otimes\ket{x_{2i}}\otimes\dots\otimes\ket{x_{ni}}\in\mathcal{H}
\end{equation*}
is a normalized product vector. Throughout this paper, the term \emph{separable}, when used without specifying a bipartition, means fully separable.
\end{definition}
Suppose that the product vectors $\{\ket{z_i}\}_{i=1}^{m}$ are linearly independent and that their span contains no product vectors other than nonzero scalar multiples of the $\ket{z_i}$. Then, after proportional terms are combined, the decomposition in \eqref{eq:separable state} is unique up to reordering and phase factors. Indeed, every product vector occurring in another decomposition of $\r$ must belong to $\mathcal{R}(\r)=\text{span}\{\ket{z_i}\}_{i=1}^{m}$. The assumption on this span determines the possible product-vector directions, while the linear independence of the operators $\{\ket{z_i}\bra{z_i}\}_{i=1}^{m}$, which follows from that of the vectors $\{\ket{z_i}\}_{i=1}^{m}$, determines the coefficients.

The following is a proposition from \cite{Chen_2013}, which summarizes some properties of three-qubit rank-four PPTES. Here and throughout, $\mathcal{R}(X)$ denotes the range of a linear operator $X$, and $\ker(X)$ denotes its kernel.
\begin{proposition}
\label{pro:summary of properties}
Let $\rho$ be a three-qubit PPTES of rank four. Then

(i) each partial transpose of $\rho$ is of rank four.

(ii) $\mathcal{R}(\rho)$ is a completely entangled space (CES), where a CES is a multipartite Hilbert subspace containing no fully product vector.

(iii) Up to a permutation of the subsystems, we have
\begin{eqnarray}
\label{eq:bipartite separable form}
\rho=\sum_{i=1}^4\ket{\alpha_i}\bra{\alpha_i}_{A_1}\otimes\ket{\psi_i}\bra{\psi_i}_{A_2A_3},
\end{eqnarray}
where any two vectors in $\{\ket{\alpha_i}\}_{i=1}^4$ are linearly independent, and the four states $\ket{\psi_i},i=1,2,3,4$ are linearly independent and entangled.

(iv)  For each bipartition, $\mathcal{R}(\rho)$ and $\ker(\rho)$ each contain exactly four bipartite product vectors up to scalar multiplication.

(v) Two normalized three-qubit PPTESs of rank four having the same range are equal.
\end{proposition}

For convenience, we give the following definition.
\begin{definition}
\label{de:characteristic set}
Let $\rho$ be a three-qubit rank-four PPTES written as \eqref{eq:bipartite separable form}, by Lemma \ref{le:four vectors}, there are invertible matrices $V_\sigma$ and $D_\sigma$, where $D_\sigma$ is diagonal, such that
\begin{eqnarray}
V_\sigma(\alpha_{\sigma(1)},\alpha_{\sigma(2)},\alpha_{\sigma(3)},\alpha_{\sigma(4)})D_\sigma=\begin{pmatrix}
    1 & 0 & 1 & t_\sigma \\
    0 & 1 & 1 & 1
\end{pmatrix},
\end{eqnarray}
where $\sigma\in S_4$. Let
\begin{eqnarray}
T_\rho:=\{t_\sigma:\sigma\in S_4\}
\end{eqnarray}
be the characteristic set of $\rho$.
\end{definition}

Below is a necessary and sufficient condition for two three-qubit rank-four PPTESs to be SLOCC-equivalent.
\begin{lemma}
\label{le:SLOCC-equivalent criterion}
Let $P,Q$ be three-qubit rank-four PPTESs such that
\begin{eqnarray}
P=\sum_{i=1}^4\ket{\alpha_i}\bra{\alpha_i}_{A_1}\otimes\ket{\psi_i}\bra{\psi_i}_{A_2A_3},\quad Q=\sum_{i=1}^4\ket{\beta_i}\bra{\beta_i}_{A_1}\otimes\ket{\phi_i}\bra{\phi_i}_{A_2A_3},
\end{eqnarray}
where $\alpha_i,\beta_i,\psi_i,\phi_i$ satisfy the conditions in (iii) of Proposition \ref{pro:summary of properties}. Then

(i) $P$ and $Q$ are SLOCC-equivalent if and only if there is an invertible product transformation $W=W_1\otimes W_2\otimes W_3$ and a permutation $\sigma$ of $\{1,2,3,4\}$ such that
\begin{eqnarray}
\label{eq:transformaton}
\ket{\alpha_i}\otimes\ket{\psi_i}=d_iW\ket{\beta_{\sigma(i)}}\otimes\ket{\phi_{\sigma(i)}},\quad i=1,2,3,4,
\end{eqnarray}
where $d_i$ is a nonzero complex number for $i=1,2,3,4$. In particular we can choose $\abs{d_i}=1$ for all $i$.

(ii) A necessary condition that $P$ is SLOCC-equivalent to $Q$ is that $P$ and $Q$ have the same characteristic set defined in Definition \ref{de:characteristic set}.
\end{lemma}
\begin{proof}
(i) If $P$ and $Q$ are SLOCC-equivalent, we have $VQV^\dagger=P$ where $V=V_1\otimes V_2\otimes V_3$ is invertible. Let
\begin{eqnarray}
\ket{\tilde{\beta}_i}=V_1\ket{\beta_i},\ \ket{\tilde{\phi}_i}=V_2\otimes V_3\ket{\phi_i},\ i=1,2,3,4,
\end{eqnarray}
then
\begin{eqnarray}
VQV^\dagger=\sum_{i=1}^4\ket{\tilde{\beta}_i}\bra{\tilde{\beta}_i}_{A_1}\otimes\ket{\tilde{\phi}_i}\bra{\tilde{\phi}_i}_{A_2A_3}=P=\sum_{i=1}^4\ket{\alpha_i}\bra{\alpha_i}_{A_1}\otimes\ket{\psi_i}\bra{\psi_i}_{A_2A_3}.
\end{eqnarray}
By (iv) of Proposition \ref{pro:summary of properties}, $\mathcal{R}(P)$ has exactly four bipartite product vectors in $\bbC^2\otimes\bbC^4$, thus \eqref{eq:transformaton} holds for $W=V$ and some permutation. Since the four vectors $\ket{\tilde{\beta}_i}\otimes\ket{\tilde{\phi}_i}$ are linearly independent, the corresponding rank-one projectors are linearly independent. Comparing the two decompositions above therefore gives $|d_i|=1$ for all $i$. Conversely if \eqref{eq:transformaton} holds, let $\rho=WQW^\dagger$, then
\begin{eqnarray}
\rho=\sum_{i=1}^4\frac{1}{\abs{d_{\sigma^{-1}(i)}}^2}\ket{\alpha_{\sigma^{-1}(i)}}\bra{\alpha_{\sigma^{-1}(i)}}_{A_1}\otimes\ket{\psi_{\sigma^{-1}(i)}}\bra{\psi_{\sigma^{-1}(i)}}_{A_2A_3}.
\end{eqnarray}
By (v) of Proposition \ref{pro:summary of properties}, $\rho$ and $P$ are equal up to scalar multiplication, thus $P$ and $Q$ are SLOCC-equivalent.

(ii) If $\eqref{eq:transformaton}$ holds then we have
\begin{eqnarray}
\ket{\alpha_i}=c_iW_1\ket{\beta_{\sigma(i)}},\ i=1,2,3,4,
\end{eqnarray}
where $c_i\neq0$, then
\begin{eqnarray}
\label{eq:alpha beta}
\frac{\det(\alpha_1,\alpha_3)\det(\alpha_2,\alpha_4)}{\det(\alpha_1,\alpha_4)\det(\alpha_2,\alpha_3)}=\frac{\det(\beta_{\sigma(1)},\beta_{\sigma(3)})\det(\beta_{\sigma(2)},\beta_{\sigma(4)})}{\det(\beta_{\sigma(1)},\beta_{\sigma(4)})\det(\beta_{\sigma(2)},\beta_{\sigma(3)})}.
\end{eqnarray}
Notice that any single element of the set \eqref{eq:set of t} determines the set, thus $P$ and $Q$ have the same characteristic set.
\end{proof}

\subsection{Lorentz invariant}
\label{sec:Lorentz invariant}

An invertible local operator has the form $G=G_1\otimes G_2\otimes\dots\otimes G_n$ with $G_k\in GL(2,\bbC)$. If $G_k\in SL(2,\bbC)$ for every $k$, then $G$ is called a local $SL(2,\bbC)$ operator and induces a special type of invertible SLOCC transformation. The following invariant plays a central role in this paper.
\begin{definition}
\label{de:Lorentz invariant}
We say that $\epsilon=
\begin{pmatrix} 0 & 1 \\
                -1 & 0
\end{pmatrix}$ is the two-dimensional Levi-Civita symbol. One can show that
\begin{eqnarray}
\forall V_0\in SL(2,\bbC),\quad V_0\epsilon V_0^{T}=V_0^T\epsilon V_0=\epsilon.
\end{eqnarray}
We refer to $\epsilon_n=\epsilon^{\otimes n}$ as the tensor product of $n$ $\epsilon$'s. For an $n$-qubit state $\r$, the quantity
\begin{equation}
    \label{eq:definition}
    I_\r:=\Tr(\rho^T\epsilon_n\rho\epsilon_n)
\end{equation}
is called the Lorentz invariant of $\r$.
\end{definition}

For square matrices $A$ and $B$ of order $2^n$, the expression
\begin{eqnarray}
\label{eq:A and B}
\Tr(A^T\epsilon_nB\epsilon_n)
\end{eqnarray}
is invariant under the product transformations $\tilde{A} = VAW,  \tilde{B} = VBW$, where $V,W\in SL(2,\bbC)^{\otimes n}$. Indeed, we can show that
\begin{eqnarray}
\label{eq:A and B under the same transformations}
\Tr(\tilde{A}^T\epsilon_n\tilde{B}\epsilon_n)=\Tr(W^TA^TV^T\epsilon_nVBW\epsilon_n)=\Tr(A^T\epsilon_nBW\epsilon_nW^T)=\Tr(A^T\epsilon_nB\epsilon_n).
\end{eqnarray}
For pure states $\rho_i=\xi_i\xi_i^\dagger$, since $\epsilon_n^T=(-1)^n\epsilon_n$, we have
\begin{eqnarray}
\label{eq:pure state's invariant}
\Tr(\rho_i^T\epsilon_n\rho_j\epsilon_n)=\Tr(\xi_i^*\xi_i^T\epsilon_n\xi_j\xi_j^\dagger\epsilon_n)=[\xi_i^T\epsilon_n\xi_j][\xi_j^\dagger\epsilon_n\xi_i^*]=(-1)^n|\xi_i^T\epsilon_n\xi_j|^2.
\end{eqnarray}
For $V\in SL(2,\bbC)^{\otimes n}$,
\begin{equation}
\label{eq:Vxi}
(V\xi_i)^T\epsilon_n(V\xi_j)=\xi_i^T\epsilon_n\xi_j.
\end{equation}
Let $\rho=\sum_{i=1}^k\lambda_i\rho_i$, where $\rho_i=\xi_i\xi_i^\dagger$ is an $n$-qubit pure state for each $i$. By \eqref{eq:definition} and \eqref{eq:pure state's invariant}, 
\begin{eqnarray}
\label{eq:Lorentz invariant}
I_{\rho}=\sum_{1\leq i,j\leq k}(-1)^n\lambda_i\lambda_j|\xi_i^T\epsilon_n\xi_j|^2.
\end{eqnarray}
Let $G=G_1\otimes G_2\otimes\dots\otimes G_n$, where $G_k\in GL(2,\bbC)$, and let $\tilde{\r}=G\r G^\dagger$. Since
\begin{equation}
    \frac{G_k}{\sqrt{\det(G_k)}}\in SL(2,\bbC),\quad k=1,2,\dots,n,
\end{equation}
by \eqref{eq:A and B under the same transformations} we have
\begin{equation}
    I_{\tilde{\r}}=\bigg|\prod_{k=1}^n\det(G_k)\bigg|^2I_\r.
\end{equation}
It is strictly invariant under local $SL(2,\bbC)$ transformations and covariant under local $GL(2,\bbC)$ transformations. Since the covariance factor is nonzero, the vanishing or nonvanishing of $I_\r$ is preserved under invertible SLOCC transformations.

For $n=3$, $I_\r$ is proportional to the quadratic Lorentz invariant $I_2(\r)$ introduced in Ref. \cite{Lorentz_invariant}, namely, $I_2(\r)=-\frac{1}{8}I_\rho$. Consequently, the two quantities yield the same zero-versus-nonzero classification. Accordingly, the terminology of Ref. \cite{Lorentz_invariant} can be expressed in terms of $I_\r$ as follows.

\begin{definition}
\label{de:type I and type II}
A three-qubit rank-four PPTES $\r$ is called type I if $I_\r\neq0$, and type II if $I_\r=0$.
\end{definition}

The following theorem characterizes type-II states in terms of the two-qubit factors in their bipartite product decompositions.

\begin{theorem}
\label{th:type II to Bell states}
Let $\rho$ be a three-qubit rank-four PPTES written as \eqref{eq:bipartite separable form}. Then $\rho$ is of type II if and only if there exist $P,Q\in GL(2,\bbC)$ and a diagonal matrix $D\in GL(4,\bbC)$ such that
\begin{eqnarray}
\label{eq:four vectors to Bell states}
(P\otimes Q)(\psi_1,\psi_2,\psi_3,\psi_4)D=(\phi_1,\phi_2,\phi_3,\phi_4)=\begin{pmatrix}
    1 & 1 & 0 & 0 \\
    0 & 0 & 1 & 1 \\
    0 & 0 & 1 & -1 \\
    1 & -1 & 0 & 0
\end{pmatrix},
\end{eqnarray}
where $D=\text{diag}(d_1,d_2,d_3,d_4)$ with $d_i\neq0$ for $i=1,2,3,4$.
\end{theorem}
The proof is given in Appendix \ref{app:proof of type II to Bell states}.

\subsection{Example}
\label{sec:example}

In this subsection, we present three examples that will be used in Sec. \ref{sec:res}: PPTESs constructed from three-qubit UPBs, a family of PPTESs that cannot be constructed from UPBs, and a type-II family.

\begin{example}
\label{ex:another representation for UPB states}
Up to local unitaries and a permutation of its members, every nontrivial three-qubit UPB can be written as follows \cite{three_qubit_UPB}:
\begin{equation}
\label{eq:three-qubit UPB}
\begin{aligned}
\{\ket{\xi_j}&=\ket{a_j,b_j,c_j}:j=1,2,3,4\}, \\
\ket{\xi_1}&=\ket{000}, \\
\ket{\xi_2}&=\ket{1}\otimes\ket{x_2}\otimes\ket{x_3}, \\
\ket{\xi_3}&=\ket{x_1}\otimes\ket{1}\otimes\ket{x'_3}, \\
\ket{\xi_4}&=\ket{x'_1}\otimes\ket{x'_2}\otimes\ket{1},
\end{aligned}
\end{equation}
where $x_k=(\cos\theta_k,\sin\theta_k)^T$, $x'_k=(-\sin\theta_k,\cos\theta_k)^T$ and $\theta_k\in(0,\frac{\pi}{2})$ for each $k$. Then
\begin{eqnarray}
\rho=\frac{1}{4}(I_8-\sum_{j=1}^4\ket{\xi_j}\bra{\xi_j})
\end{eqnarray}
is a three-qubit rank-four PPTES. Writing $\r$ in the form \eqref{eq:bipartite separable form}, and noting that $\rho$ is proportional to a rank-four projector, we obtain
\begin{eqnarray}
\label{eq:UPB and bipartite state}
\rho=\sum_{i=1}^4\frac{1}{4}\ket{\alpha_i}\bra{\alpha_i}\otimes\ket{\psi_i}\bra{\psi_i},
\end{eqnarray}
where $\alpha_i,\psi_i$ are all unit vectors. Based on the orthogonal relationship between the UPB and the range of $\rho$, i.e., $\ket{\alpha_i}\otimes\ket{\psi_i}\perp\ket{\xi_j}$ for $i,j=1,2,3,4$, we obtain that
\begin{align}
\alpha_i&\perp a_i,\quad i=1,2,3,4, \\
\psi_i&\perp b_j\otimes c_j,\quad i\neq j.
\end{align}
Then we have
\begin{eqnarray}
\ket{\alpha_1}=\ket{1},\quad\ket{\alpha_2}=\ket{0},\quad\ket{\alpha_3}=\ket{x_1'},\quad\ket{\alpha_4}=\ket{x_1}
\end{eqnarray}
and $\psi_i$ is the normalization of $e_i$, where
\begin{align*}
e_1&=(-\frac{\cos\theta_2\sin^2\theta_3}{\cos\theta_3}-\frac{\sin^2\theta_2}{\cos\theta_2\cos\theta_3},\cos\theta_2\sin\theta_3,\sin\theta_2\cos\theta_3,\sin\theta_2\sin\theta_3), \\
e_2&=(0,\cos\theta_2\sin\theta_3,\sin\theta_2\cos\theta_3,\sin\theta_2\sin\theta_3), \\
e_3&=(0,\cos\theta_2,-\frac{\sin\theta_3}{\sin\theta_2\cos\theta_3},\sin\theta_2), \\
e_4&=(0,-\frac{\sin\theta_2}{\cos\theta_2\sin\theta_3},\cos\theta_3,\sin\theta_3).
\end{align*}
\end{example}

The following family was constructed in \cite{general_position}.
\begin{example}
\label{ex:PPTES by six vectors}
Let
\begin{equation}
\label{eq:z1 to z6}
\begin{aligned}
\ket{z_1}&=\ket{e_2}\otimes(\ket{e_1}+2\ket{e_2})\otimes\ket{e_1}, \\
\ket{z_2}&=\ket{e_2}\otimes(\ket{e_1}+\ket{e_2})\otimes(\ket{e_1}+\ket{e_2}), \\
\ket{z_3}&=\ket{e_1}\otimes\ket{e_2}\otimes(\ket{e_1}-\ket{e_2}), \\
\ket{z_4}&=(\ket{e_1}+\ket{e_2})\otimes\ket{e_2}\otimes(\ket{e_1}+\ket{e_2}), \\
\ket{z_5}&=(\ket{e_1}+2\ket{e_2})\otimes\ket{e_1}\otimes\ket{e_2}, \\
\ket{z_6}&=(\ket{e_1}+\ket{e_2})\otimes(\ket{e_1}-2\ket{e_2})\otimes\ket{e_2}.
\end{aligned}
\end{equation}
The six vectors span the five-dimensional subspace $V$ that contains no other product vectors, and $\ket{e_1}\otimes\ket{e_1}\otimes\ket{e_1}$ is the unique product vector in the orthogonal complement of $V$ up to scalar multiplication. Let $\ket{\tilde{z}_i}$ be normalizations of the six product vectors and
\begin{eqnarray}
Q_p=\frac{1}{\alpha-1}(\alpha\sum_{i=1}^5p_i\ket{\tilde{z}_i}\bra{\tilde{z}_i}-\ket{\tilde{z}_6}\bra{\tilde{z}_6}),
\end{eqnarray}
where $\alpha=\frac{1}{2p_1}+\frac{2}{5p_2}+\frac{1}{5p_3}+\frac{2}{5p_4}+\frac{1}{2p_5}$ with $\sum_{i=1}^5p_i=1,p_i>0$. For each $p$, $Q_p$ is a three-qubit rank-four PPTES that cannot be constructed from UPBs \cite{general_position}.
\end{example}

The following type-II family was introduced in \cite{Lorentz_invariant}.
\begin{example}
\label{ex:PPTES=typeII}
Let $t\in\bbC\backslash\{0,1\}$, and let
\begin{align}
(b_1,b_2,b_3,b_4)&=\begin{pmatrix}
    1 & 0 & 1 & -t \\
    0 & 1 & -1 & 1
\end{pmatrix}, \\
(u_1,u_2,u_3,u_4)&=\begin{pmatrix}
    0 & -t & -t & -t \\
    1 & 0 & 1 & t \\
    -1 & 0 & 1 & t \\
    0 & 1 & -1 & -1
\end{pmatrix}.
\label{eq:u_i}
\end{align}
Let $\otimes_s$ denote the split tensor product in Definition \ref{de:split tensor product}. Then we have
\begin{eqnarray}
\label{eq:PPT with vanishing quadratic invariant}
\rho_{\text{Ref}}(t)=\alpha\sum_{i=1}^4\lambda_ib_ib_i^\dagger\otimes u_iu_i^\dagger=\alpha\sum_{i=1}^4\lambda_ib_ib_i^\dagger\otimes_su_iu_i^\dagger=\alpha\sum_{i=1}^4\lambda_iu_iu_i^\dagger\otimes b_ib_i^\dagger
\end{eqnarray}
with
$$\lambda_1=|t|^2|1-t|^2,\quad\lambda_2=|1-t|^2,\quad\lambda_3=|t|^2,\quad\lambda_4=1$$
and
$$\alpha=\frac{1}{5|t|^4+10|t|^2+1+(3|t|^2+1)|1-t|^2}.$$
For $t\in\bbC\backslash\{0,1\}$, $\r_{\text{Ref}}(t)$ is a three-qubit rank-four PPTES \cite{Lorentz_invariant}, and direct calculation gives $I_{\r_{\text{Ref}}(t)}=0$.
\end{example}

\section{Results}
\label{sec:res}

This section presents our main results in three parts. Sec. \ref{sec:type I} studies type-I three-qubit rank-four PPTESs and gives a necessary and sufficient criterion for determining whether such a state is SLOCC-equivalent to a PPTES constructed from a UPB. Sec. \ref{sec:type II} establishes an exhaustive Bell-basis normal form for type-II states and determines their SLOCC equivalence classes. Sec. \ref{sec:multi-qubit} determines the range of the Lorentz invariant for $n$-qubit states.

\subsection{Type-I PPT entangled states and UPB constructions}
\label{sec:type I}

\begin{lemma}
\label{le:invariant of upb}
Let
\begin{eqnarray}
\label{eq:rho=2x2x2pptesRank4}
\rho = \frac{1}{4}(I_8-\sum_{i=1}^4\xi_i\xi_i^{\dagger})
\end{eqnarray}
be a three-qubit rank-four PPTES constructed from a UPB as Example \ref{ex:another representation for UPB states}, where $I_8$ is the identity matrix of order eight. Let $I_\rho$ be the Lorentz invariant of $\rho$. Then the range of $I_{\rho}$ is $(-1/4,0)$. Thus, such states are of type I.
\end{lemma}
\begin{proof}
Let $\rho_i=\xi_i\xi_i^\dagger$, then $\Tr(\rho_i)=1$ for each $i$. According to \eqref{eq:definition} and \eqref{eq:pure state's invariant}, we have
\begin{equation}
\begin{aligned}
I_\rho&=\frac{1}{16}\Tr[(I_8-\sum_{i=1}^4\rho_i^T)\epsilon_3(I_8-\sum_{i=1}^4\rho_i)\epsilon_3] \\
&=\frac{1}{16}\bigg[\Tr(\epsilon_3^2)-\Tr(\sum_{i=1}^4\rho_i^T\epsilon_3^2)-\Tr(\epsilon_3\sum_{i=1}^4\rho_i\epsilon_3)+\sum_{1\leq i,j\leq4}\Tr(\rho_i^T\epsilon_3\rho_j\epsilon_3)\bigg] \\
&=\frac{1}{16}\bigg[\Tr(-I_8)-\Tr(\sum_{i=1}^4-\rho_i^T)-\Tr(-\sum_{i=1}^4\rho_i)-\sum_{1\leq i,j\leq4}|\xi_i^T\epsilon_3\xi_j|^2\bigg] \\
&=-\frac{1}{16}\sum_{1\leq i,j\leq4}|\xi_i^T\epsilon_3\xi_j|^2.
\label{eq:xi_i and xi_j}
\end{aligned}
\end{equation}
Since $\xi_i$'s are product vectors in $\bbC^2\otimes\bbC^2\otimes\bbC^2$ as \eqref{eq:three-qubit UPB}, we calculate \eqref{eq:xi_i and xi_j} and obtain
\begin{eqnarray}
\label{eq:invariant of PPTES by UPB}
I_\rho=-\frac{1}{4}(\cos^2\theta_1\cos^2\theta_2+\sin^2\theta_1\cos^2\theta_3+\sin^2\theta_2\sin^2\theta_3),\quad\theta_{1,2,3}\in(0,\frac{\pi}{2}).
\end{eqnarray}
Let
\begin{equation}
    x=\cos^2\theta_1,\quad y=\cos^2\theta_2,\quad z=\cos^2\theta_3.
\end{equation}
Then $x,y,z\in(0,1)$ and
\begin{equation}
    \begin{aligned}
        -4I_\r&=xy+(1-x)z+(1-y)(1-z) \\
        &=(1-z)[1-y(1-x)]+z[1-x(1-y)].
    \end{aligned}
\end{equation}
Thus, $0<-4I_\r<1$ and $-\frac{1}{4}<I_\r<0$. Moreover, $I_\r$ approaches $0$ as $(x,y,z)$ approaches $(0,1,0)$, and approaches $-\frac{1}{4}$ as $(x,y)$ approaches $(1,1)$. Since $I_\r$ is continuous on the connected set $(0,1)^3$, its range is $(-\frac{1}{4},0)$. Thus, $\rho$ is of type I.
\end{proof}

By the preceding lemma, every PPTES constructed from a three-qubit UPB is of type I. However, the Lorentz invariant alone does not determine the SLOCC equivalence class of a type-I state. The following example presents two type-I states, neither of which is constructed from a UPB, that have the same Lorentz invariant but are not SLOCC-equivalent.
\begin{example}
\label{ex:nonzero}
There exist two three-qubit rank-four PPTESs of type I that have the same Lorentz invariant but are not SLOCC-equivalent; neither is constructed from a UPB.

We consider three-qubit rank-four PPTESs in Example \ref{ex:PPTES by six vectors}, which cannot be constructed from UPBs. One can calculate that $Q_p=\frac{1}{\alpha-1}[0\oplus\rho]$, where
$$\rho=\left(\begin{smallmatrix}
    \frac{1}{5}\alpha p_5-\frac{1}{10} & 0 & \frac{1}{5} & 0 & \frac{2}{5}\alpha p_5-\frac{1}{10} & 0 & \frac{1}{5} \\
    0 & \frac{1}{2}\alpha p_3+\frac{1}{4}\alpha p_4 & -\frac{1}{2}\alpha p_3+\frac{1}{4}\alpha p_4 & 0 & 0 & \frac{1}{4}\alpha p_4 & \frac{1}{4}\alpha p_4 \\
    \frac{1}{5} & -\frac{1}{2}\alpha p_3+\frac{1}{4}\alpha p_4 & \frac{1}{2}\alpha p_3+\frac{1}{4}\alpha p_4-\frac{2}{5} & 0 & \frac{1}{5} & \frac{1}{4}\alpha p_4 & \frac{1}{4}\alpha p_4-\frac{2}{5} \\
    0 & 0 & 0 & \frac{1}{5}\alpha p_1+\frac{1}{4}\alpha p_2 & \frac{1}{4}\alpha p_2 & \frac{2}{5}\alpha p_1+\frac{1}{4}\alpha p_2 & \frac{1}{4}\alpha p_2 \\
    \frac{2}{5}\alpha p_5-\frac{1}{10} & 0 & \frac{1}{5} & \frac{1}{4}\alpha p_2 & \frac{1}{4}\alpha p_2+\frac{4}{5}\alpha p_5-\frac{1}{10} & \frac{1}{4}\alpha p_2 & \frac{1}{4}\alpha p_2+\frac{1}{5} \\
    0 & \frac{1}{4}\alpha p_4 & \frac{1}{4}\alpha p_4 & \frac{2}{5}\alpha p_1+\frac{1}{4}\alpha p_2 & \frac{1}{4}\alpha p_2 & \frac{4}{5}\alpha p_1+\frac{1}{4}\alpha p_2+\frac{1}{4}\alpha p_4 & \frac{1}{4}\alpha p_2+\frac{1}{4}\alpha p_4 \\
    \frac{1}{5} & \frac{1}{4}\alpha p_4 & \frac{1}{4}\alpha p_4-\frac{2}{5} & \frac{1}{4}\alpha p_2 & \frac{1}{4}\alpha p_2+\frac{1}{5} & \frac{1}{4}\alpha p_2+\frac{1}{4}\alpha p_4 & \frac{1}{4}\alpha p_2+\frac{1}{4}\alpha p_4-\frac{2}{5} \\
\end{smallmatrix}\right).$$
We calculate the Lorentz invariant of $Q_p$
\begin{align*}
I_{Q_p}=&\Tr(Q_p^T\epsilon_3Q_p\epsilon_3) \\
=&\frac{2}{(\alpha-1)^2}\big[\alpha\sum_{i=1}^5p_i|\tilde{z_i}^T\epsilon_3\tilde{z_6}|^2-\alpha^2\sum_{1\leq i<j\leq5}p_ip_j|\tilde{z_i}^T\epsilon_3\tilde{z_j}|^2\big] \\
=&\frac{2}{(\alpha-1)^2}[-\alpha^2(\frac{1}{10}p_1p_3+\frac{1}{20}p_1p_4+\frac{4}{25}p_1p_5+\frac{1}{2}p_2p_3+\frac{1}{20}p_2p_5+\frac{2}{5}p_3p_5+\frac{1}{20}p_4p_5) \\
&+\alpha(\frac{8}{25}p_1+\frac{9}{40}p_2+\frac{1}{20}p_3)],
\end{align*}
where $\alpha=\frac{1}{2p_1}+\frac{2}{5p_2}+\frac{1}{5p_3}+\frac{2}{5p_4}+\frac{1}{2p_5}$ and $\tilde{z_i}$ is the normalization of ${z_i}$ in \eqref{eq:z1 to z6}. Let
\begin{eqnarray} 
p_1\in(0,\frac{2}{5}),\quad p_5=\frac{2}{5}-p_1,\quad p_2=p_3=p_4=\frac{1}{5},
\end{eqnarray}
then we have
\begin{eqnarray}
I(p_1)=I_{Q_p}=\frac{25p_1p_5+1}{100(20p_1p_5+1)^2}(670p_1^2p_5+14p_1-800p_1^2p_5^2-277p_1p_5-12).
\end{eqnarray}
The equation $I(p_1)=-0.063$ has two solutions $p_1^{(1)},p_1^{(2)}\in(0,\frac{2}{5})$, numerically,

$$p_1^{(1)}\approx0.35651164020026005,\quad p_1^{(2)}\approx0.3943325092488642.$$

Hence, $I(p_1^{(1)})=I(p_1^{(2)})=-0.063$. Let $p_1=p_1^{(1)}$, by Appendix E in \cite{Lorentz_invariant}, we obtain four bipartite product vectors in $\mathcal{R}(Q_p)$ through numerical computation:
\begin{eqnarray}
e_i=\begin{pmatrix}
    a_i \\
    1
\end{pmatrix}\otimes\psi_i,\ i=1,2,3,4,
\end{eqnarray}
where

$$a_1\approx-8.843514882184724,\ a_2\approx0.19237140417058013,\ a_3=1,\ a_4=0.$$

By Lemma \ref{le:four vectors}, there are invertible matrices $W$ and diagonal $D$ such that
\begin{eqnarray}
W\begin{pmatrix}
    a_1 & a_2 & a_3 & a_4 \\
    1 & 1 & 1 & 1
\end{pmatrix}D=\begin{pmatrix}
    1 & 0 & 1 & t \\
    0 & 1 & 1 & 1
\end{pmatrix}.
\end{eqnarray}
Then we obtain the characteristic set of $Q_p$ as \eqref{eq:set of t}, where
\begin{eqnarray}
t=\frac{(a_1-a_3)(a_2-a_4)}{(a_1-a_4)(a_2-a_3)}\approx-0.2651270982388964.
\end{eqnarray}
Similarly, let $p_1=p_1^{(2)}$, we obtain $t'\approx-0.036855353393877174$. One can show that the characteristic sets of these two states are not the same. Thus, by (ii) of Lemma \ref{le:SLOCC-equivalent criterion}, the two three-qubit rank-four PPTESs are not SLOCC-equivalent.
\end{example}

We now give a necessary and sufficient criterion for determining whether a type-I three-qubit rank-four PPTES is SLOCC-equivalent to a state constructed from a UPB.
\begin{theorem}
\label{th:characterization of UPB-states}
Let $\r$ be a type-I three-qubit rank-four PPTES whose kernel contains four product vectors
\begin{eqnarray}
\label{eq:four product vectors}
\xi_j=a_{1j}\otimes a_{2j}\otimes a_{3j},\quad j=1,2,3,4,
\end{eqnarray}
which are in general position. By (i) of Lemma \ref{le:four vectors}, there are invertible matrices $W_k$ and diagonal $D_k$ for $k=1,2,3$, such that
\begin{eqnarray}
W_k(a_{k1},a_{k2},a_{k3},a_{k4})D_k=\begin{pmatrix}
    1 & 0 & 1 & t_k \\
    0 & 1 & 1 & 1
\end{pmatrix},\quad k=1,2,3,
\end{eqnarray}
where
\begin{eqnarray}
\label{eq:t_k}
t_k=\frac{\det(a_{k1},a_{k3})\det(a_{k2},a_{k4})}{\det(a_{k1},a_{k4})\det(a_{k2},a_{k3})}\neq0,1.
\end{eqnarray}
Let $\mathcal{A},\mathcal{B},\mathcal{C}$ be the sets \eqref{eq:A,B,C}. Then $\rho$ is SLOCC-equivalent to a PPTES constructed from a UPB if and only if $t_1,t_2,t_3\in\mathcal{A}\cup\mathcal{B}\cup\mathcal{C}$ and any two of them are not in the same set among $\mathcal{A},\mathcal{B},\mathcal{C}$.
\end{theorem}
\begin{proof}
Suppose $\rho$ is SLOCC-equivalent to a PPTES constructed from a UPB. Equivalently, there is an invertible product transformation $V=V_1\otimes V_2\otimes V_3$ such that $\{V\xi_j:j=1,2,3,4\}$ is a UPB. Without loss of generality, let
\begin{align}
V_1(a_{11},a_{12},a_{13},a_{14})&\rightarrow\begin{pmatrix}
    1 & 0 & \cos\theta_1 & -\sin\theta_1 \\
    0 & 1 & \sin\theta_1 & \cos\theta_1
\end{pmatrix}, \\
V_2(a_{21},a_{22},a_{23},a_{24})&\rightarrow\begin{pmatrix}
    1 & \cos\theta_2 & 0 & -\sin\theta_2 \\
    0 & \sin\theta_2 & 1 & \cos\theta_2
\end{pmatrix}, \\
V_3(a_{31},a_{32},a_{33},a_{34})&\rightarrow\begin{pmatrix}
    1 & \cos\theta_3 & -\sin\theta_3 & 0 \\
    0 & \sin\theta_3 & \cos\theta_3 & 1
\end{pmatrix},
\end{align}
up to scalar multiplication. By (iii) of Lemma \ref{le:four vectors}, we have
\begin{eqnarray}
t_1\in\mathcal{A},\quad t_2\in\mathcal{C},\quad t_3\in\mathcal{B}.
\end{eqnarray}
Permuting the four product vectors or the three subsystems gives the other possible orthogonality patterns. Hence the stated condition is necessary.

Conversely, suppose that $t_1,t_2,t_3$ satisfy the stated conditions. After permuting the subsystems, we may assume that $t_1\in\mathcal{A}$, $t_2\in\mathcal{C}$, and $t_3\in\mathcal{B}$. By (iii) of Lemma \ref{le:four vectors}, there exist invertible local transformations $V_k$ that produce the corresponding orthogonal pairs in each subsystem. Since $\mathcal{A},\mathcal{B},\mathcal{C}$ correspond to the three distinct pairings of four indices, the four product vectors $V\xi_j$, where $V=V_1\otimes V_2\otimes V_3$, are mutually orthogonal. They are also unextendible: otherwise, by the pigeonhole principle, one local factor of an extending product vector would be orthogonal to two linearly independent vectors in $\bbC^2$, contradicting the general-position assumption. Hence $\{V\xi_j\}_{j=1}^4$ is a UPB. Let $G=(V^\dagger)^{-1}$. Then $\ker(G\rho G^\dagger)=V\ker\rho$. By (v) of Proposition \ref{pro:summary of properties}, after normalization, $G\r G^\dagger$ is the PPTES constructed from this UPB. Therefore, $\r$ is SLOCC-equivalent to a PPTES constructed from a UPB.
\end{proof}
Notice that if the kernel of a type-I three-qubit rank-four PPTES contains no four product vectors in general position, then the state cannot be SLOCC-equivalent to a PPTES constructed from a UPB.

\subsection{Characterization of type-II PPT entangled states}
\label{sec:type II}

This subsection gives an exhaustive classification of type-II three-qubit rank-four PPTESs up to SLOCC equivalence. We first construct a Bell-basis family and determine exactly when its members are PPT entangled. We then prove that every type-II state is SLOCC-equivalent to a member of this family and characterize the SLOCC equivalence relation among its parameter values.

\begin{lemma}
\label{le:PPTES by Bell states}
Let
\begin{align}
(a_1,a_2,a_3,a_4)&=\begin{pmatrix}
    1 & 0 & 1 & t \\
    0 & 1 & 1 & 1
\end{pmatrix}, \\
(\phi_1,\phi_2,\phi_3,\phi_4)&=\begin{pmatrix}
    1 & 1 & 0 & 0 \\
    0 & 0 & 1 & 1 \\
    0 & 0 & 1 & -1 \\
    1 & -1 & 0 & 0
\end{pmatrix},
\end{align}
where $t\in\bbC\backslash\{0,1\}$. Suppose
\begin{eqnarray}
\label{eq:ei}
\rho=\sum_{i=1}^4\lambda_ie_ie_i^\dagger,\quad e_i=a_i\otimes\phi_i,\quad\lambda_i>0.
\end{eqnarray}
Then, after normalization, $\r$ is a PPTES if and only if there exists $c>0$ such that
\begin{equation}
\label{eq:proportion}
(\lambda_1,\lambda_2,\lambda_3,\lambda_4)=c(|t^2-t|,|t-1|,|t|,1).
\end{equation}
\end{lemma}
The proof is given in Appendix \ref{app:conjugate congruent transformations}.

Suppressing the normalization factor, define the one-parameter family $\r(t)$ as follows:
\begin{align}
\rho(t)=&|t^2-t|(a_1a_1^\dagger)\otimes(\phi_1\phi_1^\dagger)+|t-1|(a_2a_2^\dagger)\otimes(\phi_2\phi_2^\dagger)+|t|(a_3a_3^\dagger)\otimes(\phi_3\phi_3^\dagger) \nonumber\\
&+(a_4a_4^\dagger)\otimes(\phi_4\phi_4^\dagger) \nonumber\\
\label{eq:PPTES by Bell states}
=&\begin{pmatrix}
    |t^2-t| & 0 & 0 & |t^2-t| & 0 & 0 & 0 & 0 \\
    0 & |t|+|t|^2 & |t|-|t|^2 & 0 & 0 & |t|+t & |t|-t & 0 \\
    0 & |t|-|t|^2 & |t|+|t|^2 & 0 & 0 & |t|-t & |t|+t & 0 \\
    |t^2-t| & 0 & 0 & |t^2-t| & 0 & 0 & 0 & 0 \\
    0 & 0 & 0 & 0 & |t-1| & 0 & 0 & -|t-1| \\
    0 & |t|+\bar{t} & |t|-\bar{t} & 0 & 0 & |t|+1 & |t|-1 & 0 \\
    0 & |t|-\bar{t} & |t|+\bar{t} & 0 & 0 & |t|-1 & |t|+1 & 0 \\
    0 & 0 & 0 & 0 & -|t-1| & 0 & 0 & |t-1|
\end{pmatrix},
\end{align}
where $t\neq0,1$ is a complex number. By Lemma \ref{le:PPTES by Bell states}, it is a family of PPTESs, and direct calculation gives $I_{\r(t)}=0$.

Since the kernel of $\rho(t)$ is the orthogonal complement of the subspace generated by $e_i$'s, we obtain $\ker\rho(t)$ as the following four-dimensional linear subspace in $\bbC^8$,
\begin{align}
\label{eq:kernel of rho}
\{ &(a,-\frac{(|t|^2+t)c+(|t|^2-t)d}{2|t|^2},-\frac{(|t|^2-t)c+(|t|^2+t)d}{2|t|^2},-a,b,c,d,b): \notag \\
&(a,b,c,d)\in\bbC^4 \}.
\end{align}
Assume $x=(x_1,x_2,x_3,x_4,x_5,x_6,x_7,x_8)$ in \eqref{eq:kernel of rho} is a product vector in $\bbC^2\otimes\bbC^2\otimes\bbC^2$, then by $x_2x_7=x_3x_6$ we have $c=\pm d$ since $|t|^2-t\neq0$. By $x_1x_8=x_4x_5$ we have $a=0$ or $b=0$. If $b=0$, by $x_5x_8=x_6x_7$ and $c=\pm d$ we have $c=d=0$, then $a=0$ and $x$ is a zero vector. If $a=0$, by $x_1x_4=x_2x_3$ and $c=\pm d$ we have $c=d=0$, then $b=0$ and $x$ is a zero vector. Thus, the kernel of $\rho(t)$ contains no product vector in $\bbC^2\otimes\bbC^2\otimes\bbC^2$.

By either Lemma \ref{le:invariant of upb} or Theorem \ref{th:characterization of UPB-states}, we can conclude that $\r(t)$ cannot be SLOCC-equivalent to a PPTES constructed from a UPB.

\begin{theorem}
\label{th:PPTES of type II}
Every type-II three-qubit rank-four PPTES is SLOCC-equivalent to the Bell-basis normal form $\r(t)$ in \eqref{eq:PPTES by Bell states}, for some $t\in\bbC\backslash\{0,1\}$.
\end{theorem}
\begin{proof}
According to Proposition \ref{pro:summary of properties}, we write a three-qubit rank-four PPTES $\varrho$ of type II as
\begin{eqnarray}
\label{eq:alpha psi}
\varrho=\sum_{i=1}^4\ket{\alpha_i}\bra{\alpha_i}\otimes\ket{\psi_i}\bra{\psi_i}.
\end{eqnarray}
By (i) of Lemma \ref{le:four vectors}, there are invertible matrices $V$ and diagonal $\Lambda$ such that
\begin{eqnarray}
V(\alpha_1,\alpha_2,\alpha_3,\alpha_4)\Lambda=(a_1,a_2,a_3,a_4)=\begin{pmatrix}
    1 & 0 & 1 & t \\
    0 & 1 & 1 & 1
\end{pmatrix}.
\end{eqnarray}
By Theorem \ref{th:type II to Bell states}, there are invertible matrices $W=P\otimes Q$ and diagonal $D$ such that \eqref{eq:four vectors to Bell states} holds. Then we have
\begin{eqnarray}
\label{eq:di}
a_i\otimes\phi_i=d_i(V\otimes P\otimes Q)(\alpha_i\otimes\psi_i),\quad d_i\neq0,\quad i=1,2,3,4,
\end{eqnarray}
where $\phi_i$'s are the same as in \eqref{eq:PPTES by Bell states}. Let $G=V\otimes P\otimes Q$, then by \eqref{eq:alpha psi} and \eqref{eq:di}
\begin{equation}
G\varrho G^\dagger=\sum_{i=1}^4\frac{1}{|d_i|^2}\ket{a_i\otimes\phi_i}\bra{a_i\otimes\phi_i}.
\end{equation}
Since $G\varrho G^\dagger$ is, up to normalization, a three-qubit rank-four PPTES, by Lemma \ref{le:PPTES by Bell states},
\begin{equation}
(\frac{1}{|d_1|^2},\frac{1}{|d_2|^2},\frac{1}{|d_3|^2},\frac{1}{|d_4|^2})=c(|t^2-t|,|t-1|,|t|,1),\quad c>0.
\end{equation}
Thus, $G\varrho G^\dagger=c\r(t)$, proving that $\varrho$ is SLOCC-equivalent to $\r(t)$ for some $t\in\bbC\backslash\{0,1\}$.
\end{proof}

In the following, we further classify three-qubit rank-four PPTESs of type II in terms of SLOCC equivalence.
\begin{corollary}
\label{co:classify type II}
Let $\r(t),\r(t')$ be type-II PPTESs in \eqref{eq:PPTES by Bell states}, where $t,t'\in\bbC\backslash\{0,1\}$. Then $\rho(t)$ and $\r(t')$ are SLOCC-equivalent if and only if $t'$ belongs to the set \eqref{eq:set of t}. Equivalently, they have the same characteristic set.
\end{corollary}
\begin{proof}
Necessity follows from (ii) of Lemma \ref{le:SLOCC-equivalent criterion}. Conversely, if $t'$ belongs to the set \eqref{eq:set of t}, let
\begin{eqnarray}
(a_1,a_2,a_3,a_4)=\begin{pmatrix}
    1 & 0 & 1 & t \\
    0 & 1 & 1 & 1
\end{pmatrix},\quad(a_1',a_2',a_3',a_4')=\begin{pmatrix}
    1 & 0 & 1 & t' \\
    0 & 1 & 1 & 1
\end{pmatrix}.
\end{eqnarray}
Then by (ii) of Lemma \ref{le:four vectors}, there is a permutation $\sigma\in S_4$, invertible matrices $W_1$ and diagonal $D$ such that
\begin{eqnarray}
\label{eq:W1}
W_1(a_1',a_2',a_3',a_4')D=(a_{\sigma(1)},a_{\sigma(2)},a_{\sigma(3)},a_{\sigma(4)}).
\end{eqnarray}
By Lemma \ref{le:permutations of Bell states} and \eqref{eq:W1}, there are unitary matrices $W_2$ and $W_3$ such that
\begin{eqnarray}
(W_1\otimes W_2\otimes W_3)(a_i'\otimes\phi_i)=d_i(a_{\sigma(i)}\otimes\phi_{\sigma(i)}),\quad d_i\neq0,\quad i=1,2,3,4.
\end{eqnarray}
Thus, by (i) of Lemma \ref{le:SLOCC-equivalent criterion}, $\rho(t)$ and $\rho(t')$ are SLOCC-equivalent.
\end{proof}

Ref. \cite{Lorentz_invariant} introduced the one-parameter family $\r_{\text{Ref}}(t)$ in \eqref{eq:PPT with vanishing quadratic invariant}, which reproduced all type-II states found numerically but was not proved to be exhaustive. Let
\begin{eqnarray}
W_1=\begin{pmatrix}
    1 & 0 \\
    0 & -1
\end{pmatrix},\quad W_2=\begin{pmatrix}
    \frac{\sqrt{t}-1}{\sqrt{1-t}} & \frac{\sqrt{t}-t}{\sqrt{1-t}} \\
    1 & \sqrt{t}
\end{pmatrix},\quad W_3=\begin{pmatrix}
    1 & \sqrt{t} \\
    \frac{1-\sqrt{t}}{\sqrt{1-t}} & \frac{t-\sqrt{t}}{\sqrt{1-t}}
\end{pmatrix}
\end{eqnarray}
and $V=W_1\otimes W_2\otimes W_3$. Direct substitution gives $V\r_{\text{Ref}}(t)V^\dagger\propto\r(t)$. Thus, for each $t\in\bbC\backslash\{0,1\}$, $\r_{\text{Ref}}(t)$ and $\r(t)$ are SLOCC-equivalent. Together with Theorem \ref{th:PPTES of type II}, this shows that every type-II three-qubit rank-four PPTES is SLOCC-equivalent to $\r_{\text{Ref}}(t)$ for some $t\in\bbC\backslash\{0,1\}$.

Combining the preceding results, we obtain the following decision procedure for determining the type of a three-qubit rank-four PPTES $\r$ and whether it is SLOCC-equivalent to a PPTES constructed from a UPB. We first calculate the Lorentz invariant $I_\rho$ and classify it as one of the two types, namely type I and type II. If $I_\rho=0$, then $\rho$ is SLOCC-equivalent to $\r(t)$ in \eqref{eq:PPTES by Bell states}, for some $t\in\bbC\backslash\{0,1\}$. If $I_\rho\neq0$, then $\r$ is of type I, and we consider $\ker\r$. Using Appendix E of \cite{Lorentz_invariant}, for a fixed bipartition we can determine the four bipartite product vectors in $\ker\r$ guaranteed by (iv) of Proposition \ref{pro:summary of properties}, each up to a nonzero scalar multiple. Then we determine whether they are tripartite product vectors in general position. If they are, let \eqref{eq:four product vectors} be the four product vectors in any order and calculate \eqref{eq:t_k}. If $t_1,t_2,t_3$ satisfy the conditions in Theorem \ref{th:characterization of UPB-states}, then $\rho$ is SLOCC-equivalent to a PPTES constructed from a UPB. In all other cases, no such SLOCC equivalence exists. See Figure \ref{fig:flowchart}.

\begin{figure}[htbp]
\centering
\begin{tikzpicture}[node distance=3cm, auto,
    decision/.style={diamond, draw, fill=blue!10, text width=7em, text centered, inner sep=0.5pt, font=\normalsize},
    block/.style={rectangle, draw, fill=blue!05, text width=7.7em, text centered, rounded corners, inner sep=2pt, font=\normalsize},
    arrow/.style={thick, ->, >=stealth}]

\node[block] (PPTES) {Three-qubit rank-four PPTES $\rho$};
\node[decision, below of=PPTES, node distance=2.5cm] (invariant) {Compute $I_\rho$};
\node[block, below of=invariant] (zero) {$\rho$ is of type II};
\node[block, below of=zero] (well) {$\rho$ is SLOCC-equivalent to $\r(t)$ in \eqref{eq:PPTES by Bell states} for some $t$};
\node[block, right of=invariant, node distance=5cm] (nonzero) {$\rho$ is of type I};
\node[block, right of=nonzero, node distance=3.7cm] (kernel) {For a fixed bipartition, find the four bipartite product vectors in $\ker\rho$};
\node[decision, below of=kernel, node distance=3.5cm] (prod) {Are they tripartite product vectors in general position?};
\node[block, below of=prod, node distance=3.5cm] (calc) {Compute $t_1,t_2,t_3$};
\node[decision, below of=calc, node distance=2.5cm] (cond) {Do $t_1,t_2,t_3$ satisfy the conditions in Theorem \ref{th:characterization of UPB-states}?};
\node[block, below of=cond, node distance=3.5cm] (UPB) {SLOCC-equivalent to a PPTES constructed from a UPB};
\node[block, right of=cond, node distance=5cm] (notUPB) {Not SLOCC-equivalent to any PPTES constructed from a UPB};

\draw[arrow] (PPTES)--(invariant);
\draw[arrow] (invariant)--node[left] {$I_\rho=0$} (zero);
\draw[arrow] (zero)--(well);
\draw[arrow] (invariant)--node[above] {$I_\rho\neq0$} (nonzero);
\draw[arrow] (nonzero)--(kernel);
\draw[arrow] (kernel)--(prod);
\draw[arrow] (prod)--node[left] {Yes} (calc);
\draw[arrow] (prod)--node[above, fill=white, inner sep=1pt] {No} (notUPB);
\draw[arrow] (calc)--(cond);
\draw[arrow] (cond)--node[left] {Yes} (UPB);
\draw[arrow] (cond)--node[above] {No} (notUPB);
\end{tikzpicture}
\caption{Decision procedure for determining the type of a three-qubit rank-four PPTES and its SLOCC equivalence to a UPB-derived PPTES.}
\label{fig:flowchart}
\end{figure}

\subsection{\texorpdfstring{$n$}{n}-qubit states}
\label{sec:multi-qubit}

In this subsection, we analyze the Lorentz invariant of multiqubit states and obtain the explicit range. Throughout this subsection, we assume $n\geq2$.

\begin{lemma}
\label{le:zero invariant}
An $n$-qubit state $\rho$ with zero Lorentz invariant has rank at most $2^{n-1}$. In particular, a three-qubit PPTES with zero Lorentz invariant must have rank four.
\end{lemma}
\begin{proof}
Suppose $\rank\rho=k$, then $\rho$ can be written as
\begin{eqnarray}
\rho=\sum_{i=1}^k\xi_i\xi_i^\dagger,
\end{eqnarray}
where $\xi_i$'s are linearly independent. By \eqref{eq:Lorentz invariant},
\begin{eqnarray}
I_\rho=\Tr(\rho^T\epsilon_n\rho\epsilon_n)=\sum_{1\leq i,j\leq k}(-1)^n|\xi_i^T\epsilon_n\xi_j|^2=0.
\end{eqnarray}
Then we have
\begin{eqnarray}
|\xi_i^T\epsilon_n\xi_j|=0,\quad 1\leq i,j\leq k.
\end{eqnarray}
Thus, the two \(k\)-dimensional subspaces spanned by $\{\bar{\xi}_i\}_{i=1}^k$ and $\{\epsilon_n\xi_j\}_{j=1}^k$, respectively, are orthogonal. Hence $\rank\rho=k\leq2^{n-1}$.

For $n=3$, the preceding bound gives $\rank\r\leq4$. Since every multipartite PPT state of rank at most three is fully separable \cite{Chen_2013}, a three-qubit PPTES with $I_\r=0$ must have rank four.
\end{proof}

\begin{lemma}
\label{le:pure state's invariant}
Let $\rho=\xi\xi^\dagger$ be an $n$-qubit pure state, where $\xi$ is normalized. The Lorentz invariant $I_\rho$ is identically equal to zero for odd $n$, and its range is $[0,1]$ for even $n$.
\end{lemma}
\begin{proof}
By \eqref{eq:Lorentz invariant}, we have
\begin{eqnarray}
I_\rho=(-1)^n|\xi^T\epsilon_n\xi|^2.
\end{eqnarray}
If $n$ is odd, $\epsilon_n$ is skew-symmetric. Then
\begin{eqnarray}
\xi^T\epsilon_n\xi=(\xi^T\epsilon_n\xi)^T=\xi^T\epsilon_n^T\xi=-\xi^T\epsilon_n\xi,
\end{eqnarray}
which means that $I_\rho=0$.

On the other hand if $n$ is even, then Cauchy–Schwarz inequality implies that
\begin{eqnarray}
|\xi^T\epsilon_n\xi|\leq\|\xi\|^2=1.
\end{eqnarray}
Let $\xi$ be the normalized vector whose first and last entries are $1/\sqrt{2}$ and whose remaining entries are zero. Then $|\xi^T\epsilon_n\xi|=1$. In addition, we choose $\xi=\alpha\otimes\psi$ where $\psi\in\bbC^{2^{n-1}}$. Since $n-1$ is odd,
\begin{eqnarray}
\xi^T\epsilon_n\xi=(\alpha\otimes\psi)^T(\epsilon\otimes\epsilon_{n-1})(\alpha\otimes\psi)=(\alpha^T\epsilon\alpha)(\psi^T\epsilon_{n-1}\psi)=0.
\end{eqnarray}
Since the set of normalized state vectors is connected and $I_\r$ depends continuously on $\xi$, its range is $[0,1]$.
\end{proof}

\begin{lemma}
\label{le:rank-two}
Let $\rho=c_1\xi_1\xi_1^\dagger+c_2\xi_2\xi_2^\dagger$ be an $n$-qubit mixed state of rank two, where $c_1+c_2=1$, $c_1,c_2>0$ and $\xi_1,\xi_2$ are normalized. Let $I_\rho$ be the Lorentz invariant of $\rho$.

(i) For odd $n$, the range of $I_\rho$ is $[-1/2,0]$.

(ii) For even $n$, the range of $I_\rho$ is $[0,1)$.
\end{lemma}
\begin{proof}
(i) By \eqref{eq:Lorentz invariant} and Lemma \ref{le:pure state's invariant}, we have
\begin{eqnarray}
I_\rho=-c_1^2|\xi_1^T\epsilon_n\xi_1|^2-2c_1c_2|\xi_1^T\epsilon_n\xi_2|^2-c_2^2|\xi_2^T\epsilon_n\xi_2|^2=-2c_1c_2|\xi_1^T\epsilon_n\xi_2|^2.
\end{eqnarray}
Therefore,
\begin{equation}
    -\frac{1}{2}\leq-2c_1c_2\leq I_\r\leq0.
\end{equation}
It remains to show that every value in this interval is attained. Since $n$ is odd and $n\geq2$, we have $n\geq3$. For any $s\in[0,1]$, let
\begin{equation}
    c_1=c_2=\frac12,\qquad\xi_1=\ket{0}^{\otimes n},
\end{equation}
and
\begin{equation}
    \xi_2=\sqrt{s}\ket{1}^{\otimes n}+\sqrt{1-s}\ket{10\cdots0}.
\end{equation}
The vectors $\xi_1$ and $\xi_2$ are orthonormal, and
\begin{equation}
    |\xi_1^T\epsilon_n\xi_2|^2=s.
\end{equation}
Hence the corresponding rank-two state satisfies
\begin{equation}
    I_\r=-\frac{s}{2}.
\end{equation}
As $s$ ranges over $[0,1]$, $I_\r$ ranges over $[-1/2,0]$, proving (i).

(ii) By \eqref{eq:Lorentz invariant}, we have
\begin{eqnarray}
I_\rho=c_1^2|\xi_1^T\epsilon_n\xi_1|^2+2c_1c_2|\xi_1^T\epsilon_n\xi_2|^2+c_2^2|\xi_2^T\epsilon_n\xi_2|^2.
\end{eqnarray}
Since $\epsilon_n$ is unitary and the vectors $\xi_i$ are normalized, the Cauchy–Schwarz inequality gives
\begin{equation}
|\xi_i^T\epsilon_n\xi_j|=|\langle\bar{\xi}_i,\epsilon_n\xi_j\rangle|\leq1.
\end{equation}
Therefore, $I_\r=1$ would imply
\begin{eqnarray}
\label{eq:invariant is equal to 1}
|\xi_1^T\epsilon_n\xi_1|=|\xi_1^T\epsilon_n\xi_2|=|\xi_2^T\epsilon_n\xi_2|=1,
\end{eqnarray}
Hence $\xi_1$ and $\xi_2$ are parallel, contradicting the assumption that $\r$ has rank two. Thus $I_\r<1$. Let $\rho=(\alpha\alpha^\dagger)\otimes\rho'$, where $\rho'$ is an $(n-1)$-qubit state of rank two, then we have
\begin{eqnarray}
I_\rho=|\alpha^T\epsilon\alpha|^2\Tr({\rho^\prime}^T\epsilon_{n-1}\rho^\prime\epsilon_{n-1})=0.
\end{eqnarray}
The set of rank-two states is connected, and $I_\r$ is continuous on this set. We have shown that $0$ is attained whereas $1$ is not. Moreover, rank-two states can approach a pure state with $I_\r=1$, whose existence is guaranteed by Lemma \ref{le:pure state's invariant}, so the supremum is $1$. Therefore, the range of $I_\r$ is $[0,1)$.
\end{proof}

The preceding lemmas establish rigorous results for states of ranks one and two. We now consider states of arbitrary rank.
\begin{theorem}
    \label{th:arbitrary rank}
    For a fixed integer $n\geq2$, the range of the Lorentz invariant $I_\r$ over all $n$-qubit states $\r$ is $[-\frac12,0]$ for odd $n$ and $[0,1]$ for even $n$.
\end{theorem}
\begin{proof}
    Let
    \begin{equation}
        \r=\sum_{i=1}^r\lambda_i\xi_i\xi_i^\dagger,\quad\sum_{i=1}^r\lambda_i=1
    \end{equation}
    be a spectral decomposition of $\r$, where $r=\rank\r$ and the vectors $\xi_i$ are orthonormal.
    
    Suppose that $n$ is odd. Set
    \begin{equation}
        \label{eq:m_ij}
        m_{ij}=|\xi_i^T\epsilon_n\xi_j|^2.
    \end{equation}
    Since $\epsilon_n$ is skew-symmetric for odd $n$, $m_{ii}=0$ and $m_{ij}=m_{ji}$. Moreover, since $\epsilon_n$ is unitary, the vectors $\epsilon_n\xi_j$ are orthonormal. Hence, by Bessel’s inequality,
    \begin{equation}
        \sum_{j=1}^rm_{ij}=\sum_{j=1}^r|\langle\bar{\xi}_i,\epsilon_n\xi_j\rangle|^2\leq1.
    \end{equation}
    For $i\neq j$, since $\lambda_i+\lambda_j\leq1$, we have
    \begin{equation}
        4\lambda_i\lambda_j\leq(\lambda_i+\lambda_j)^2\leq\lambda_i+\lambda_j.
    \end{equation}
    By \eqref{eq:Lorentz invariant} and \eqref{eq:m_ij},
    \begin{equation}
    \begin{aligned}
        I_\r&=-\sum_{1\leq i,j\leq r}\lambda_i\lambda_jm_{ij} \\
        &=-2\sum_{i<j}\lambda_i\lambda_jm_{ij} \\
        &\geq-\frac{1}{2}\sum_{i<j}(\lambda_i+\lambda_j)m_{ij} \\
        &=-\frac{1}{2}\sum_{i=1}^r\bigg(\lambda_i\sum_{j\neq i}m_{ij}\bigg)\geq-\frac{1}{2}.
    \end{aligned}
    \end{equation}
    Thus, $-\frac{1}{2}\leq I_\r\leq0$. By Lemma \ref{le:rank-two}, every value in this interval is attained by a rank-two state. Hence the range is $[-\frac{1}{2},0]$.
    
    If $n$ is even, \eqref{eq:Lorentz invariant} and the Cauchy–Schwarz inequality imply
    \begin{equation}
        0\leq I_\r\leq\sum_{i,j}\lambda_i\lambda_j=1.
    \end{equation}
    By Lemma \ref{le:pure state's invariant}, every value in $[0,1]$ is already attained by a pure state. Hence the range is $[0,1]$.
\end{proof}

\section{Conclusion}
\label{sec:con}

We studied three-qubit rank-four PPTESs using the Lorentz invariant. We proved that every such state constructed from a three-qubit UPB is of type I and established a necessary and sufficient criterion for determining whether a type-I state is SLOCC-equivalent to a state constructed from a UPB. For type-II states, we derived a Bell-basis normal form, proved that it exhausts the entire class up to SLOCC equivalence, and determined exactly when two members of this normal form are SLOCC-equivalent. In particular, this provides the exhaustiveness result that was not proved for the parametrized family in Ref. \cite{Lorentz_invariant}. A complete SLOCC classification of type-I states that are not SLOCC-equivalent to PPTESs constructed from UPBs remains open.

We also determined the range of the Lorentz invariant over $n$-qubit states of arbitrary rank. It is $[-\frac{1}{2},0]$ for odd $n$ and $[0,1]$ for even $n$. Investigating the role of Lorentz invariants in broader classifications of multipartite PPTESs is a natural direction for future work.

\section*{ACKNOWLEDGMENTS}

This work is supported by the NNSF of China (Grant No. 12471427).

\appendix

\section{Proof of Lemma \ref{le:four vectors}}
\label{app:proof of four vectors}
\begin{proof}
(i) Since $ad-bc,cf-de,ah-bg\neq0$, we have
\begin{eqnarray}
\begin{pmatrix}
    d & -c \\
    -b & a
\end{pmatrix}x\text{diag}\Big(\frac{1}{ad-bc},\frac{1}{ad-bc},\frac{1}{de-cf},\frac{1}{ah-bg}\Big)=y=\begin{pmatrix}
    1 & 0 & 1 & t_2 \\
    0 & 1 & t_1 & 1
\end{pmatrix},
\end{eqnarray}
where $t_1=\frac{af-be}{de-cf},t_2=\frac{dg-ch}{ah-bg}$. Since $\det(y_i,y_j)\neq0$ for $i\neq j$, we have
\begin{eqnarray}
t:=t_1t_2=\frac{\det(y_1,y_3)\det(y_2,y_4)}{\det(y_1,y_4)\det(y_2,y_3)}=\frac{\det(x_1,x_3)\det(x_2,x_4)}{\det(x_1,x_4)\det(x_2,x_3)}\neq0,1.
\end{eqnarray}
Then after multiplication by
\begin{eqnarray}
\begin{pmatrix}
    t_1 & 0 \\
    0 & 1
\end{pmatrix}
\end{eqnarray}
and normalization we obtain \eqref{eq:standard form}.

(ii) Let
\begin{eqnarray}
(a_1,a_2,a_3,a_4)=\begin{pmatrix}
    1 & 0 & 1 & t \\
    0 & 1 & 1 & 1
\end{pmatrix},
\end{eqnarray}
then we have
\begin{eqnarray}
\label{eq:t_sigma}
t_\sigma=\frac{\det(x_{\sigma(1)},x_{\sigma(3)})\det(x_{\sigma(2)},x_{\sigma(4)})}{\det(x_{\sigma(1)},x_{\sigma(4)})\det(x_{\sigma(2)},x_{\sigma(3)})}=\frac{\det(a_{\sigma(1)},a_{\sigma(3)})\det(a_{\sigma(2)},a_{\sigma(4)})}{\det(a_{\sigma(1)},a_{\sigma(4)})\det(a_{\sigma(2)},a_{\sigma(3)})}.
\end{eqnarray}
We calculate \eqref{eq:t_sigma} for all $\sigma$ and obtain \eqref{eq:set of t}.

(iii) If $t<0$, then we have
\begin{eqnarray}
\text{diag}(1,\sqrt{-t})\begin{pmatrix}
    1 & 0 & 1 & t \\
    0 & 1 & 1 & 1
\end{pmatrix}=\begin{pmatrix}
    1 & 0 & 1 & t \\
    0 & \sqrt{-t} & \sqrt{-t} & \sqrt{-t}
\end{pmatrix},
\end{eqnarray}
where the first two columns are orthogonal and the last two columns are orthogonal. Since the diagonal matrix in (i) only rescales the columns, $V=\text{diag}(1,\sqrt{-t})W$ gives $Vx_1\perp Vx_2$ and $Vx_3\perp Vx_4$. If $0<t<1$, use the ordering $(x_1,x_4,x_2,x_3)$. By (ii), the corresponding parameter is

$$t_\sigma=1-\frac1t<0.$$
Applying the preceding argument gives $Vx_1\perp Vx_4$ and $Vx_2\perp Vx_3$. If $t>1$, use the ordering $(x_1,x_3,x_2,x_4)$, for which

$$t_\sigma=1-t<0.$$
The same argument gives $Vx_1\perp Vx_3$ and $Vx_2\perp Vx_4$. This proves the sufficiency and the orthogonality relations in \eqref{eq:t in ABC}.

Conversely, suppose that such an invertible transformation exists. After permuting the four vectors, we may assume that the two orthogonal pairs are $(x_{\sigma(1)},x_{\sigma(2)})$ and $(x_{\sigma(3)},x_{\sigma(4)})$. By (ii), since the nonzero column rescalings in $D_\sigma$ do not affect orthogonality, there is an invertible matrix such that the vectors $z_i$ defined by
\begin{eqnarray}
(z_1,z_2,z_3,z_4)=\begin{pmatrix}
    w_1 & w_2 \\
    w_3 & w_4
\end{pmatrix}\begin{pmatrix}
    1 & 0 & 1 & t_\sigma \\
    0 & 1 & 1 & 1
\end{pmatrix}
\end{eqnarray}
satisfy $z_1\perp z_2$ and $z_3\perp z_4$. Then
\begin{align*}
\langle z_1,z_2\rangle&=\bar{w_1}w_2+\bar{w_3}w_4=0, \\
\langle z_3,z_4\rangle&=(|w_1|^2+|w_3|^2)t_\sigma+|w_2|^2+|w_4|^2+(w_1\bar{w}_2+w_3\bar{w}_4)t_\sigma+\bar{w}_1w_2+\bar{w}_3w_4=0.
\end{align*}
Since the first equality and its complex conjugate eliminate the cross terms, the second equality implies
\begin{equation}
    t_\sigma=-\frac{|w_2|^2+|w_4|^2}{|w_1|^2+|w_3|^2}<0.
\end{equation}
It follows from \eqref{eq:set of t} that $t\in\bbR$. Since $t\neq0,1$, we have $t\in\mathcal{A}\cup\mathcal{B}\cup\mathcal{C}$.
\end{proof}

\section{Proof of Theorem \ref{th:type II to Bell states}}
\label{app:proof of type II to Bell states}
\begin{proof}
By \eqref{eq:Lorentz invariant}, we have
\begin{eqnarray}
I_\rho=-\sum_{1\leq i,j\leq4}|\alpha_i^T\epsilon\alpha_j|^2|\psi_i^T(\epsilon\otimes\epsilon)\psi_j|^2.
\end{eqnarray}
Let $\alpha_i=(a_{1i},a_{2i})^T$, then
\begin{eqnarray}
\label{eq:collinearity criterion}
\alpha_i^T\epsilon\alpha_j=a_{1i}a_{2j}-a_{2i}a_{1j}.
\end{eqnarray}
Since any two of $\alpha_i$'s are linearly independent, \eqref{eq:collinearity criterion} is equal to zero only for $i=j$. Thus, $I_\rho=0$ if and only if
\begin{eqnarray}
\label{eq:criterion for type II}
\psi_i^T(\epsilon\otimes\epsilon)\psi_j=0,\quad\forall1\leq i<j\leq4.
\end{eqnarray}
We only need to prove that \eqref{eq:criterion for type II} holds if and only if \eqref{eq:four vectors to Bell states} holds.

We first prove the sufficiency (the "if" part). If \eqref{eq:four vectors to Bell states} holds, we have
\begin{equation}
    \psi_i=\frac{1}{d_i}(P\otimes Q)^{-1}\phi_i,\quad i=1,2,3,4.
\end{equation}
For any $P,Q\in GL(2,\bbC)$, since
\begin{equation}
    \label{eq:PTPQTQ}
    R^T\epsilon R=R\epsilon R^T=\det(R)\epsilon,\quad R\in GL(2,\bbC),
\end{equation}
we obtain that
\begin{equation}
    \psi_i^T(\epsilon\otimes\epsilon)\psi_j=\frac{\phi_i^T(\epsilon\otimes\epsilon)\phi_j}{d_id_j\det(P)\det(Q)}=0.
\end{equation}
Hence \eqref{eq:criterion for type II} holds, and therefore $\r$ is of type II.

We next prove the necessity (the "only if" part). If \eqref{eq:criterion for type II} holds, let the bijection
\begin{eqnarray}
\label{eq:F}
F:(x,y,z,w)^T\mapsto F(x,y,z,w)=\begin{pmatrix}
    x & y \\
    z & w
\end{pmatrix}
\end{eqnarray}
and
\begin{eqnarray}
P\otimes Q=\begin{pmatrix}
    p_{11} & p_{12} \\
    p_{21} & p_{22}
\end{pmatrix}\otimes\begin{pmatrix}
    q_{11} & q_{12} \\
    q_{21} & q_{22}
\end{pmatrix},\quad\psi=\begin{pmatrix}
    x_1 \\
    x_2 \\
    x_3 \\
    x_4
\end{pmatrix}.
\end{eqnarray}
Then we have
\begin{eqnarray}
\label{eq:vectors to matrices}
(P\otimes Q)\psi=\begin{pmatrix}
    p_{11}q_{11}x_1+p_{11}q_{12}x_2+p_{12}q_{11}x_3+p_{12}q_{12}x_4 \\
    p_{11}q_{21}x_1+p_{11}q_{22}x_2+p_{12}q_{21}x_3+p_{12}q_{22}x_4 \\
    p_{21}q_{11}x_1+p_{21}q_{12}x_2+p_{22}q_{11}x_3+p_{22}q_{12}x_4 \\
    p_{21}q_{21}x_1+p_{21}q_{22}x_2+p_{22}q_{21}x_3+p_{22}q_{22}x_4
\end{pmatrix}=F^{-1}(PF(\psi)Q^T).
\end{eqnarray}
Let
\begin{eqnarray}
\label{eq:A_i} 
A_i=F(\psi_i),\quad i=1,2,3,4,
\end{eqnarray}
which are invertible since each $\psi_i$ is entangled. Since for $\psi,\phi\in\bbC^4$,
\begin{equation}
\Tr\big(F(\psi)^TF(\phi)\big)=\psi^T\phi,
\end{equation}
by \eqref{eq:criterion for type II} and \eqref{eq:vectors to matrices} we have
\begin{eqnarray}
\label{eq:zero trace}
\Tr(A_i^T\epsilon A_j\epsilon)=\psi_i^TF^{-1}\big(\epsilon F(\psi_j)\epsilon\big)=\psi_i^T(\epsilon\otimes\epsilon^T)\psi_j=0,\quad i\neq j.
\end{eqnarray}
By \eqref{eq:PTPQTQ} and \eqref{eq:zero trace}, we have
\begin{equation}
    \label{eq:invariant trace}
    \Tr\big[(PA_iQ)^T\epsilon(PA_jQ)\epsilon\big]=\det(P)\det(Q)\Tr(A_i^T\epsilon A_j\epsilon)=0,\quad i\neq j.
\end{equation}
Therefore, the zero condition in \eqref{eq:zero trace} is preserved under arbitrary invertible transformations
\begin{equation}
    A_i\mapsto PA_iQ.
\end{equation}
There are invertible matrices $P_0$ and $Q_0$ such that $P_0A_1Q_0=I_2$. Let
\begin{eqnarray}
\label{eq:A_i'}
A_i'=P_0A_iQ_0,\quad i=1,2,3,4.
\end{eqnarray}
Since the Jordan canonical form of an invertible complex $2\times2$ matrix is either diagonal or a single Jordan block, there exists an invertible matrix $P_1$ such that
\begin{eqnarray}
\label{eq:diagonal or Jordan}
P_1A_2'P_1^{-1}=\begin{pmatrix}
    s_1 & 0 \\
    0 & s_2
\end{pmatrix}\quad or\quad\begin{pmatrix}
    s_0 & 1 \\
    0 & s_0
\end{pmatrix},
\end{eqnarray}
where $s_0,s_1,s_2\neq0$. Let
\begin{eqnarray}
\label{eq:A_i''}    
A_i''=P_1A_i'P_1^{-1}=\begin{pmatrix}
    a_i & b_i \\
    c_i & d_i
\end{pmatrix},\quad i=1,2,3,4,
\end{eqnarray}
then $A_1''=I_2$. By \eqref{eq:invariant trace}, we have
\begin{eqnarray}
\label{eq:A_1'' and A_i''}
\Tr(A_1''^T\epsilon A_i''\epsilon)=-\Tr(A_i'')=0,\quad i=2,3,4.
\end{eqnarray}
Thus, $\Tr(A_2'')=0$. If $A_2''$ is of the Jordan form in \eqref{eq:diagonal or Jordan}, we obtain $2s_0=0$, contradicting $s_0\neq0$. Hence $A_2''$ is diagonal with zero trace and can be written as
\begin{eqnarray}
A_2''=s_1\begin{pmatrix}
    1 & 0 \\
    0 & -1
\end{pmatrix}.
\end{eqnarray}
By \eqref{eq:A_1'' and A_i''} and \eqref{eq:invariant trace}, we have
\begin{equation}
\begin{aligned}
\Tr(A_i'')&=a_i+d_i=0,\quad &i=3,4, \\
\Tr(A_2''^T\epsilon A_i''\epsilon)&=s_1(a_i-d_i)=0,\quad &i=3,4, \\
\Tr(A_3''^T\epsilon A_4''\epsilon)&=b_3c_4+c_3b_4-a_3d_4-d_3a_4=0.
\end{aligned}
\end{equation}
Solving equations above yields that
\begin{eqnarray}
a_3=a_4=d_3=d_4=0,\quad\frac{c_3}{b_3}+\frac{c_4}{b_4}=0,
\end{eqnarray}
where $b_3,b_4,c_3,c_4\neq0$ since $A_3'',A_4''$ are invertible. Let $s=\frac{c_3}{b_3}$, we can write $A_i''$ as follows,
\begin{eqnarray}
\label{eq:concrete A_i''}
A_1''=I_2,\quad A_2''=s_1\begin{pmatrix}
    1 & 0 \\
    0 & -1
\end{pmatrix},\quad A_3''=b_3\begin{pmatrix}
    0 & 1 \\
    s & 0
\end{pmatrix},\quad A_4''=b_4\begin{pmatrix}
    0 & 1 \\
    -s & 0
\end{pmatrix}.
\end{eqnarray}
Let $P_2=\text{diag}(\sqrt{s},1),\ Q_2=\text{diag}(1,\sqrt{s})$ and
\begin{eqnarray}
\label{eq:A_i'''}
A_i'''=P_2A_i''Q_2,\quad i=1,2,3,4,
\end{eqnarray}
where $\sqrt{s}$ is the principal square root of $s$. Then we have
\begin{eqnarray}
\label{eq:concrete A_i'''}
A_1'''=\sqrt{s}I_2,\quad A_2'''=s_1\sqrt{s}\begin{pmatrix}
    1 & 0 \\
    0 & -1
\end{pmatrix},\quad A_3'''=b_3s\begin{pmatrix}
    0 & 1 \\
    1 & 0
\end{pmatrix},\quad A_4'''=b_4s\begin{pmatrix}
    0 & 1 \\
    -1 & 0
\end{pmatrix}.
\end{eqnarray}
Let $P=P_2P_1P_0$ and $Q^T=Q_0P_1^{-1}Q_2$, by \eqref{eq:A_i'}, \eqref{eq:A_i''} and \eqref{eq:A_i'''}, we have
\begin{eqnarray}
\label{eq:PA_iQ}
PA_iQ^T=A_i''',\quad i=1,2,3,4.
\end{eqnarray}
We apply the mapping $F^{-1}$ to both sides of \eqref{eq:PA_iQ}, then by \eqref{eq:vectors to matrices},
\begin{eqnarray}
(P\otimes Q)\psi_i=F^{-1}(PA_iQ^T)=F^{-1}(A_i''').
\end{eqnarray}
Thus, let $D=\text{diag}(l_1,l_2,l_3,l_4)$, where
\begin{eqnarray}
l_1=\frac{1}{\sqrt{s}},\quad l_2=\frac{1}{s_1\sqrt{s}},\quad l_3=\frac{1}{b_3s},\quad l_4=\frac{1}{b_4s},
\end{eqnarray}
we obtain \eqref{eq:four vectors to Bell states}.
\end{proof}

\section{Proof of Lemma \ref{le:PPTES by Bell states}}
\label{app:conjugate congruent transformations}
\begin{proof}
The four product vectors $e_i$'s in \eqref{eq:ei} are in general position in $\bbC^2\otimes\bbC^4$. According to Theorem \ref{th:GP}, the span of $\{e_i\}$ has no other bipartite product vectors in $\bbC^2\otimes\bbC^4$ up to scalar multiplication. Thus, $\rho$ has a unique decomposition into pure bipartite product states in $\bbC^2\otimes\bbC^4$. Since $\{e_i\}$ are not tripartite product vectors, $\rho$ is not fully separable and hence is a three-qubit rank-four entangled state. Notice that $\rho^{T_1}$ is positive semidefinite and $\rho^{T_2}=\rho^{T_3}$. So $\rho$ is a PPTES if and only if
\begin{eqnarray}
\label{eq:rhoT2>=0}    
\rho^{T_2}\geq0.
\end{eqnarray}
Since $\lambda_4>0$, after dividing $\r$ by $\lambda_4$, we may assume without loss of generality that $\lambda_4=1$ and $\lambda_{1,2,3}>0$. Then we have
\begin{eqnarray}
\rho^{T_2}=\left(\begin{smallmatrix}
    \lambda_1 & 0 & 0 & \lambda_3-|t|^2 & 0 & 0 & 0 & \lambda_3-t \\
    0 & \lambda_3+|t|^2 & \lambda_1 & 0 & 0 & \lambda_3+t & 0 & 0 \\
    0 & \lambda_1 & \lambda_3+|t|^2 & 0 & 0 & 0 & \lambda_3+t & 0 \\
    \lambda_3-|t|^2 & 0 & 0 & \lambda_1 & \lambda_3-t & 0 & 0 & 0 \\
    0 & 0 & 0 & \lambda_3-\bar{t} & \lambda_2 & 0 & 0 & \lambda_3-1 \\
    0 & \lambda_3+\bar{t} & 0 & 0 & 0 & \lambda_3+1 & -\lambda_2 & 0 \\
    0 & 0 & \lambda_3+\bar{t} & 0 & 0 & -\lambda_2 & \lambda_3+1 & 0 \\
    \lambda_3-\bar{t} & 0 & 0 & 0 & \lambda_3-1 & 0 & 0 & \lambda_2
\end{smallmatrix}\right).
\end{eqnarray}
Applying the following invertible congruence transformations, where each column operation is accompanied by the conjugate row operation and each interchange is performed simultaneously on the corresponding row and column, we obtain:

Column~\textcircled{4}-$\frac{\lambda_3-|t|^2}{\lambda_1}$\textcircled{1} and row~\textcircled{4}-$\frac{\lambda_3-|t|^2}{\lambda_1}$\textcircled{1}:

$$\left(\begin{smallmatrix}
    \lambda_1 & 0 & 0 & 0 & 0 & 0 & 0 & \lambda_3-t \\
    0 & \lambda_3+|t|^2 & \lambda_1 & 0 & 0 & \lambda_3+t & 0 & 0 \\
    0 & \lambda_1 & \lambda_3+|t|^2 & 0 & 0 & 0 & \lambda_3+t & 0 \\
    0 & 0 & 0 & \lambda_1-\frac{(\lambda_3-|t|^2)^2}{\lambda_1} & \lambda_3-t & 0 & 0 & -\frac{(\lambda_3-|t|^2)(\lambda_3-t)}{\lambda_1} \\
    0 & 0 & 0 & \lambda_3-\bar{t} & \lambda_2 & 0 & 0 & \lambda_3-1 \\
    0 & \lambda_3+\bar{t} & 0 & 0 & 0 & \lambda_3+1 & -\lambda_2 & 0 \\
    0 & 0 & \lambda_3+\bar{t} & 0 & 0 & -\lambda_2 & \lambda_3+1 & 0 \\
    \lambda_3-\bar{t} & 0 & 0 & -\frac{(\lambda_3-|t|^2)(\lambda_3-\bar{t})}{\lambda_1} & \lambda_3-1 & 0 & 0 & \lambda_2
\end{smallmatrix}\right)$$

Column~\textcircled{8}-$\frac{\lambda_3-t}{\lambda_1}$\textcircled{1} and row~\textcircled{8}-$\frac{\lambda_3-\bar{t}}{\lambda_1}$\textcircled{1}:

$$\left(\begin{smallmatrix}
    \lambda_1 & 0 & 0 & 0 & 0 & 0 & 0 & 0 \\
    0 & \lambda_3+|t|^2 & \lambda_1 & 0 & 0 & \lambda_3+t & 0 & 0 \\
    0 & \lambda_1 & \lambda_3+|t|^2 & 0 & 0 & 0 & \lambda_3+t & 0 \\
    0 & 0 & 0 & \lambda_1-\frac{(\lambda_3-|t|^2)^2}{\lambda_1} & \lambda_3-t & 0 & 0 & -\frac{(\lambda_3-|t|^2)(\lambda_3-t)}{\lambda_1} \\
    0 & 0 & 0 & \lambda_3-\bar{t} & \lambda_2 & 0 & 0 & \lambda_3-1 \\
    0 & \lambda_3+\bar{t} & 0 & 0 & 0 & \lambda_3+1 & -\lambda_2 & 0 \\
    0 & 0 & \lambda_3+\bar{t} & 0 & 0 & -\lambda_2 & \lambda_3+1 & 0 \\
    0 & 0 & 0 & -\frac{(\lambda_3-|t|^2)(\lambda_3-\bar{t})}{\lambda_1} & \lambda_3-1 & 0 & 0 & \lambda_2-\frac{(\lambda_3-t)(\lambda_3-\bar{t})}{\lambda_1}
\end{smallmatrix}\right)$$

Omitting the isolated $1\times1$ block $[\lambda_1]$, we obtain the remaining $7\times7$ block:

$$\left(\begin{smallmatrix}
    \lambda_3+|t|^2 & \lambda_1 & 0 & 0 & \lambda_3+t & 0 & 0 \\
    \lambda_1 & \lambda_3+|t|^2 & 0 & 0 & 0 & \lambda_3+t & 0 \\
    0 & 0 & \lambda_1-\frac{(\lambda_3-|t|^2)^2}{\lambda_1} & \lambda_3-t & 0 & 0 & -\frac{(\lambda_3-|t|^2)(\lambda_3-t)}{\lambda_1} \\
    0 & 0 & \lambda_3-\bar{t} & \lambda_2 & 0 & 0 & \lambda_3-1 \\
    \lambda_3+\bar{t} & 0 & 0 & 0 & \lambda_3+1 & -\lambda_2 & 0 \\
    0 & \lambda_3+\bar{t} & 0 & 0 & -\lambda_2 & \lambda_3+1 & 0 \\
    0 & 0 & -\frac{(\lambda_3-|t|^2)(\lambda_3-\bar{t})}{\lambda_1} & \lambda_3-1 & 0 & 0 & \lambda_2-\frac{(\lambda_3-t)(\lambda_3-\bar{t})}{\lambda_1}
\end{smallmatrix}\right)$$

Column~\textcircled{2}-$\frac{\lambda_1}{\lambda_3+|t|^2}$\textcircled{1} and row~\textcircled{2}-$\frac{\lambda_1}{\lambda_3+|t|^2}$\textcircled{1}:

$$\left(\begin{smallmatrix}
    \lambda_3+|t|^2 & 0 & 0 & 0 & \lambda_3+t & 0 & 0 \\
    0 & \lambda_3+|t|^2-\frac{\lambda_1^2}{\lambda_3+|t|^2} & 0 & 0 & -\frac{\lambda_1(\lambda_3+t)}{\lambda_3+|t|^2} & \lambda_3+t & 0 \\
    0 & 0 & \lambda_1-\frac{(\lambda_3-|t|^2)^2}{\lambda_1} & \lambda_3-t & 0 & 0 & -\frac{(\lambda_3-|t|^2)(\lambda_3-t)}{\lambda_1} \\
    0 & 0 & \lambda_3-\bar{t} & \lambda_2 & 0 & 0 & \lambda_3-1 \\
    \lambda_3+\bar{t} & -\frac{\lambda_1(\lambda_3+\bar{t})}{\lambda_3+|t|^2} & 0 & 0 & \lambda_3+1 & -\lambda_2 & 0 \\
    0 & \lambda_3+\bar{t} & 0 & 0 & -\lambda_2 & \lambda_3+1 & 0 \\
    0 & 0 & -\frac{(\lambda_3-|t|^2)(\lambda_3-\bar{t})}{\lambda_1} & \lambda_3-1 & 0 & 0 & \lambda_2-\frac{(\lambda_3-t)(\lambda_3-\bar{t})}{\lambda_1}
\end{smallmatrix}\right)$$

Column~\textcircled{5}-$\frac{\lambda_3+t}{\lambda_3+|t|^2}$\textcircled{1} and row~\textcircled{5}-$\frac{\lambda_3+\bar{t}}{\lambda_3+|t|^2}$\textcircled{1}:

$$\left(\begin{smallmatrix}
    \lambda_3+|t|^2 & 0 & 0 & 0 & 0 & 0 & 0 \\
    0 & \lambda_3+|t|^2-\frac{\lambda_1^2}{\lambda_3+|t|^2} & 0 & 0 & -\frac{\lambda_1(\lambda_3+t)}{\lambda_3+|t|^2} & \lambda_3+t & 0 \\
    0 & 0 & \lambda_1-\frac{(\lambda_3-|t|^2)^2}{\lambda_1} & \lambda_3-t & 0 & 0 & -\frac{(\lambda_3-|t|^2)(\lambda_3-t)}{\lambda_1} \\
    0 & 0 & \lambda_3-\bar{t} & \lambda_2 & 0 & 0 & \lambda_3-1 \\
    0 & -\frac{\lambda_1(\lambda_3+\bar{t})}{\lambda_3+|t|^2} & 0 & 0 & \frac{\lambda_3(1-t)(1-\bar{t})}{\lambda_3+|t|^2} & -\lambda_2 & 0 \\
    0 & \lambda_3+\bar{t} & 0 & 0 & -\lambda_2 & \lambda_3+1 & 0 \\
    0 & 0 & -\frac{(\lambda_3-|t|^2)(\lambda_3-\bar{t})}{\lambda_1} & \lambda_3-1 & 0 & 0 & \lambda_2-\frac{(\lambda_3-t)(\lambda_3-\bar{t})}{\lambda_1}
\end{smallmatrix}\right)$$

Omitting the isolated $1\times1$ block $[\lambda_3+|t|^2]$, we obtain the remaining $6\times6$ block:

$$\left(\begin{smallmatrix}
    \lambda_3+|t|^2-\frac{\lambda_1^2}{\lambda_3+|t|^2} & 0 & 0 & -\frac{\lambda_1(\lambda_3+t)}{\lambda_3+|t|^2} & \lambda_3+t & 0 \\
    0 & \lambda_1-\frac{(\lambda_3-|t|^2)^2}{\lambda_1} & \lambda_3-t & 0 & 0 & -\frac{(\lambda_3-|t|^2)(\lambda_3-t)}{\lambda_1} \\
    0 & \lambda_3-\bar{t} & \lambda_2 & 0 & 0 & \lambda_3-1 \\
    -\frac{\lambda_1(\lambda_3+\bar{t})}{\lambda_3+|t|^2} & 0 & 0 & \frac{\lambda_3(1-t)(1-\bar{t})}{\lambda_3+|t|^2} & -\lambda_2 & 0 \\
    \lambda_3+\bar{t} & 0 & 0 & -\lambda_2 & \lambda_3+1 & 0 \\
    0 & -\frac{(\lambda_3-|t|^2)(\lambda_3-\bar{t})}{\lambda_1} & \lambda_3-1 & 0 & 0 & \lambda_2-\frac{(\lambda_3-t)(\lambda_3-\bar{t})}{\lambda_1}
\end{smallmatrix}\right)$$

Swap \textcircled{3} and \textcircled{1}:

$$\left(\begin{smallmatrix}
    \lambda_2 & \lambda_3-\bar{t} & 0 & 0 & 0 & \lambda_3-1 \\
    \lambda_3-t & \lambda_1-\frac{(\lambda_3-|t|^2)^2}{\lambda_1} & 0 & 0 & 0 & -\frac{(\lambda_3-|t|^2)(\lambda_3-t)}{\lambda_1} \\
    0 & 0 & \lambda_3+|t|^2-\frac{\lambda_1^2}{\lambda_3+|t|^2} & -\frac{\lambda_1(\lambda_3+t)}{\lambda_3+|t|^2} & \lambda_3+t & 0 \\
    0 & 0 & -\frac{\lambda_1(\lambda_3+\bar{t})}{\lambda_3+|t|^2} & \frac{\lambda_3(1-t)(1-\bar{t})}{\lambda_3+|t|^2} & -\lambda_2 & 0 \\
    0 & 0 & \lambda_3+\bar{t} & -\lambda_2 & \lambda_3+1 & 0 \\
    \lambda_3-1 & -\frac{(\lambda_3-|t|^2)(\lambda_3-\bar{t})}{\lambda_1} & 0 & 0 & 0 & \lambda_2-\frac{(\lambda_3-t)(\lambda_3-\bar{t})}{\lambda_1}
\end{smallmatrix}\right)$$

Column~\textcircled{2}-$\frac{\lambda_3-\bar{t}}{\lambda_2}$\textcircled{1} and row~\textcircled{2}-$\frac{\lambda_3-t}{\lambda_2}$\textcircled{1}:

$$\left(\begin{smallmatrix}
    \lambda_2 & 0 & 0 & 0 & 0 & \lambda_3-1 \\
    0 & \lambda_1-\frac{(\lambda_3-|t|^2)^2}{\lambda_1}-\frac{(\lambda_3-t)(\lambda_3-\bar{t})}{\lambda_2} & 0 & 0 & 0 & \frac{(\lambda_3-t)(1-\lambda_3)}{\lambda_2}-\frac{(\lambda_3-|t|^2)(\lambda_3-t)}{\lambda_1} \\
    0 & 0 & \lambda_3+|t|^2-\frac{\lambda_1^2}{\lambda_3+|t|^2} & -\frac{\lambda_1(\lambda_3+t)}{\lambda_3+|t|^2} & \lambda_3+t & 0 \\
    0 & 0 & -\frac{\lambda_1(\lambda_3+\bar{t})}{\lambda_3+|t|^2} & \frac{\lambda_3(1-t)(1-\bar{t})}{\lambda_3+|t|^2} & -\lambda_2 & 0 \\
    0 & 0 & \lambda_3+\bar{t} & -\lambda_2 & \lambda_3+1 & 0 \\
    \lambda_3-1 & \frac{(\lambda_3-\bar{t})(1-\lambda_3)}{\lambda_2}-\frac{(\lambda_3-|t|^2)(\lambda_3-\bar{t})}{\lambda_1} & 0 & 0 & 0 & \lambda_2-\frac{(\lambda_3-t)(\lambda_3-\bar{t})}{\lambda_1}
\end{smallmatrix}\right)$$

Column~\textcircled{6}-$\frac{\lambda_3-1}{\lambda_2}$\textcircled{1} and row~\textcircled{6}-$\frac{\lambda_3-1}{\lambda_2}$\textcircled{1}:

$$\left(\begin{smallmatrix}
    \lambda_2 & 0 & 0 & 0 & 0 & 0 \\
    0 & \lambda_1-\frac{(\lambda_3-|t|^2)^2}{\lambda_1}-\frac{(\lambda_3-t)(\lambda_3-\bar{t})}{\lambda_2} & 0 & 0 & 0 & \frac{(\lambda_3-t)(1-\lambda_3)}{\lambda_2}-\frac{(\lambda_3-|t|^2)(\lambda_3-t)}{\lambda_1} \\
    0 & 0 & \lambda_3+|t|^2-\frac{\lambda_1^2}{\lambda_3+|t|^2} & -\frac{\lambda_1(\lambda_3+t)}{\lambda_3+|t|^2} & \lambda_3+t & 0 \\
    0 & 0 & -\frac{\lambda_1(\lambda_3+\bar{t})}{\lambda_3+|t|^2} & \frac{\lambda_3(1-t)(1-\bar{t})}{\lambda_3+|t|^2} & -\lambda_2 & 0 \\
    0 & 0 & \lambda_3+\bar{t} & -\lambda_2 & \lambda_3+1 & 0 \\
    0 & \frac{(\lambda_3-\bar{t})(1-\lambda_3)}{\lambda_2}-\frac{(\lambda_3-|t|^2)(\lambda_3-\bar{t})}{\lambda_1} & 0 & 0 & 0 & \lambda_2-\frac{(\lambda_3-t)(\lambda_3-\bar{t})}{\lambda_1}-\frac{(\lambda_3-1)^2}{\lambda_2}
\end{smallmatrix}\right)$$

Omitting the isolated $1\times1$ block $[\lambda_2]$, we obtain the remaining $5\times5$ block:

$$\left(\begin{smallmatrix}
    \lambda_1-\frac{(\lambda_3-|t|^2)^2}{\lambda_1}-\frac{(\lambda_3-t)(\lambda_3-\bar{t})}{\lambda_2} & 0 & 0 & 0 & \frac{(\lambda_3-t)(1-\lambda_3)}{\lambda_2}-\frac{(\lambda_3-|t|^2)(\lambda_3-t)}{\lambda_1} \\
    0 & \lambda_3+|t|^2-\frac{\lambda_1^2}{\lambda_3+|t|^2} & -\frac{\lambda_1(\lambda_3+t)}{\lambda_3+|t|^2} & \lambda_3+t & 0 \\
    0 & -\frac{\lambda_1(\lambda_3+\bar{t})}{\lambda_3+|t|^2} & \frac{\lambda_3(1-t)(1-\bar{t})}{\lambda_3+|t|^2} & -\lambda_2 & 0 \\
    0 & \lambda_3+\bar{t} & -\lambda_2 & \lambda_3+1 & 0 \\
    \frac{(\lambda_3-\bar{t})(1-\lambda_3)}{\lambda_2}-\frac{(\lambda_3-|t|^2)(\lambda_3-\bar{t})}{\lambda_1} & 0 & 0 & 0 & \lambda_2-\frac{(\lambda_3-t)(\lambda_3-\bar{t})}{\lambda_1}-\frac{(\lambda_3-1)^2}{\lambda_2}
\end{smallmatrix}\right)$$

Swap \textcircled{4} and \textcircled{1}:

$$\left(\begin{smallmatrix}
    \lambda_3+1 & \lambda_3+\bar{t} & -\lambda_2 & 0 & 0 \\
    \lambda_3+t & \lambda_3+|t|^2-\frac{\lambda_1^2}{\lambda_3+|t|^2} & -\frac{\lambda_1(\lambda_3+t)}{\lambda_3+|t|^2} & 0 & 0 \\
    -\lambda_2 & -\frac{\lambda_1(\lambda_3+\bar{t})}{\lambda_3+|t|^2} & \frac{\lambda_3(1-t)(1-\bar{t})}{\lambda_3+|t|^2} & 0 & 0 \\
    0 & 0 & 0 & \lambda_1-\frac{(\lambda_3-|t|^2)^2}{\lambda_1}-\frac{(\lambda_3-t)(\lambda_3-\bar{t})}{\lambda_2} & \frac{(\lambda_3-t)(1-\lambda_3)}{\lambda_2}-\frac{(\lambda_3-|t|^2)(\lambda_3-t)}{\lambda_1} \\
    0 & 0 & 0 & \frac{(\lambda_3-\bar{t})(1-\lambda_3)}{\lambda_2}-\frac{(\lambda_3-|t|^2)(\lambda_3-\bar{t})}{\lambda_1} & \lambda_2-\frac{(\lambda_3-t)(\lambda_3-\bar{t})}{\lambda_1}-\frac{(\lambda_3-1)^2}{\lambda_2}
\end{smallmatrix}\right)$$

Written as $M\oplus\begin{pmatrix}
    \lambda_1-\frac{(\lambda_3-|t|^2)^2}{\lambda_1}-\frac{(\lambda_3-t)(\lambda_3-\bar{t})}{\lambda_2} & \frac{(\lambda_3-t)(1-\lambda_3)}{\lambda_2}-\frac{(\lambda_3-|t|^2)(\lambda_3-t)}{\lambda_1} \\
    \frac{(\lambda_3-\bar{t})(1-\lambda_3)}{\lambda_2}-\frac{(\lambda_3-|t|^2)(\lambda_3-\bar{t})}{\lambda_1} & \lambda_2-\frac{(\lambda_3-t)(\lambda_3-\bar{t})}{\lambda_1}-\frac{(\lambda_3-1)^2}{\lambda_2}
\end{pmatrix}$

$$M=\begin{pmatrix}
    \lambda_3+1 & \lambda_3+\bar{t} & -\lambda_2\\
    \lambda_3+t & \lambda_3+|t|^2-\frac{\lambda_1^2}{\lambda_3+|t|^2} & -\frac{\lambda_1(\lambda_3+t)}{\lambda_3+|t|^2}\\
    -\lambda_2 & -\frac{\lambda_1(\lambda_3+\bar{t})}{\lambda_3+|t|^2} & \frac{\lambda_3(1-t)(1-\bar{t})}{\lambda_3+|t|^2}
\end{pmatrix}$$

Column~\textcircled{2}-$\frac{\lambda_3+\bar{t}}{\lambda_3+1}$\textcircled{1} and row~\textcircled{2}-$\frac{\lambda_3+t}{\lambda_3+1}$\textcircled{1}:

$$\begin{pmatrix}
    \lambda_3+1 & 0 & -\lambda_2\\
    0 & \lambda_3+|t|^2-\frac{\lambda_1^2}{\lambda_3+|t|^2}-\frac{(\lambda_3+t)(\lambda_3+\bar{t})}{\lambda_3+1} & \frac{\lambda_2(\lambda_3+t)}{\lambda_3+1}-\frac{\lambda_1(\lambda_3+t)}{\lambda_3+|t|^2}\\
    -\lambda_2 & \frac{\lambda_2(\lambda_3+\bar{t})}{\lambda_3+1}-\frac{\lambda_1(\lambda_3+\bar{t})}{\lambda_3+|t|^2} & \frac{\lambda_3(1-t)(1-\bar{t})}{\lambda_3+|t|^2}
\end{pmatrix}$$

Column~\textcircled{3}+$\frac{\lambda_2}{\lambda_3+1}$\textcircled{1} and row~\textcircled{3}+$\frac{\lambda_2}{\lambda_3+1}$\textcircled{1}:

$$\begin{pmatrix}
    \lambda_3+1 & 0 & 0\\
    0 & \lambda_3+|t|^2-\frac{\lambda_1^2}{\lambda_3+|t|^2}-\frac{(\lambda_3+t)(\lambda_3+\bar{t})}{\lambda_3+1} & \frac{\lambda_2(\lambda_3+t)}{\lambda_3+1}-\frac{\lambda_1(\lambda_3+t)}{\lambda_3+|t|^2}\\
    0 & \frac{\lambda_2(\lambda_3+\bar{t})}{\lambda_3+1}-\frac{\lambda_1(\lambda_3+\bar{t})}{\lambda_3+|t|^2} & \frac{\lambda_3(1-t)(1-\bar{t})}{\lambda_3+|t|^2}-\frac{\lambda_2^2}{\lambda_3+1}
\end{pmatrix}$$
Thus there exists an invertible matrix $V$, such that $V\rho^{T_2}V^\dagger=\text{diag}(\lambda_1,\lambda_3+|t|^2,\lambda_2,\lambda_3+1)\oplus P\oplus Q$, where
\begin{eqnarray}
\label{eq:P}
P&=[p_{ij}]&=\begin{pmatrix}
    \lambda_3+|t|^2-\frac{\lambda_1^2}{\lambda_3+|t|^2}-\frac{(\lambda_3+t)(\lambda_3+\bar{t})}{\lambda_3+1} & \frac{\lambda_2(\lambda_3+t)}{\lambda_3+1}-\frac{\lambda_1(\lambda_3+t)}{\lambda_3+|t|^2}\\
    \frac{\lambda_2(\lambda_3+\bar{t})}{\lambda_3+1}-\frac{\lambda_1(\lambda_3+\bar{t})}{\lambda_3+|t|^2} & \frac{\lambda_3(1-t)(1-\bar{t})}{\lambda_3+|t|^2}-\frac{\lambda_2^2}{\lambda_3+1}
\end{pmatrix} \\
\label{eq:Q}
Q&=[q_{ij}]&=\begin{pmatrix}
    \lambda_1-\frac{(\lambda_3-|t|^2)^2}{\lambda_1}-\frac{(\lambda_3-t)(\lambda_3-\bar{t})}{\lambda_2} & \frac{(\lambda_3-t)(1-\lambda_3)}{\lambda_2}-\frac{(\lambda_3-|t|^2)(\lambda_3-t)}{\lambda_1} \\
    \frac{(\lambda_3-\bar{t})(1-\lambda_3)}{\lambda_2}-\frac{(\lambda_3-|t|^2)(\lambda_3-\bar{t})}{\lambda_1} & \lambda_2-\frac{(\lambda_3-t)(\lambda_3-\bar{t})}{\lambda_1}-\frac{(\lambda_3-1)^2}{\lambda_2}
\end{pmatrix}.
\end{eqnarray}

If $\r$ is a PPTES, by (i) of Proposition \ref{pro:summary of properties}, $\rank\rho^{T_2}=4$. Thus $P$ and $Q$ are zero matrices, then
\begin{equation}
\label{eq:zero matrix}
p_{11}=p_{12}=p_{22}=q_{11}=q_{12}=q_{22}=0.
\end{equation}
If $\lambda_3+t=0$, then $t<0$. From $p_{11}=p_{22}=0$ we have $\lambda_1=t^2-t,\lambda_2=1-t$. Thus solving \eqref{eq:zero matrix} yields that
\begin{eqnarray}
A:=\{(\lambda_1,\lambda_2,\lambda_3,t)\in\bbR_+^3\times\bbC:\lambda_1=t^2-t,\lambda_2=1-t,\lambda_3=-t,t<0\}.
\end{eqnarray}
If $\lambda_3-t=0$, then $t>0$. From $q_{11}=q_{22}=0$ we have $\lambda_1=|t^2-t|,\lambda_2=|t-1|$. Thus solving \eqref{eq:zero matrix} yields that
\begin{eqnarray}
B:=\{(\lambda_1,\lambda_2,\lambda_3,t)\in\bbR_+^3\times\bbC:\lambda_1=|t^2-t|,\lambda_2=|t-1|,\lambda_3=t,t>0,t\neq1\}.
\end{eqnarray}
If $\lambda_3\pm t\neq0$, then $p_{12}=0$ means $\frac{\lambda_2}{\lambda_3+1}=\frac{\lambda_1}{\lambda_3+|t|^2}$. We substitute them into $p_{11}=0$ and obtain
\begin{equation}
\label{eq:lambda123 and t}
\lambda_1\lambda_2=\lambda_3(1-t)(1-\bar{t}).
\end{equation}
From $p_{12}=0$ and $q_{12}=0$, we obtain $\lambda_3=|t|,\lambda_1=\lambda_2|t|$, which, substituted into \eqref{eq:lambda123 and t}, yields $\lambda_2=|t-1|$. Excluding the case that $\lambda_3\pm t=0$, we obtain
\begin{eqnarray}
C:=\{(\lambda_1,\lambda_2,\lambda_3,t)\in\bbR_+^3\times\bbC:\lambda_1=|t^2-t|,\lambda_2=|t-1|,\lambda_3=|t|,t\notin\bbR\}.
\end{eqnarray}
In summary, we have $(\lambda_1,\lambda_2,\lambda_3,t)\in D=A\cup B\cup C$, where
\begin{eqnarray}
D=\{(\lambda_1,\lambda_2,\lambda_3,t)\in\bbR_+^3\times\bbC:\lambda_1=|t^2-t|,\lambda_2=|t-1|,\lambda_3=|t|,t\neq0,1\}.
\end{eqnarray}
Restoring the overall positive factor $c=\lambda_4$, we obtain \eqref{eq:proportion}.

Conversely if \eqref{eq:proportion} holds, direct substitution of the coefficients in \eqref{eq:proportion} into \eqref{eq:P} and \eqref{eq:Q} gives $P=Q=0$. Since the remaining four diagonal entries of $V\rho^{T_2}V^\dagger$ are positive, it follows that $\rho^{T_2}\geq0$. Consequently, $\r$ is a PPTES. Thus, we complete the proof.
\end{proof}

\bibliographystyle{unsrt}
\bibliography{references}

@article{length,
  title = {Entanglement Detection Length of Multipartite Quantum States},
  author = {Shi, Fei and Chen, Lin and Chiribella, Giulio and Zhao, Qi},
  journal = {Phys. Rev. Lett.},
  volume = {134},
  issue = {5},
  pages = {050201},
  numpages = {6},
  year = {2025},
  month = {Feb},
  publisher = {American Physical Society},
  doi = {10.1103/PhysRevLett.134.050201},
  url = {https://link.aps.org/doi/10.1103/PhysRevLett.134.050201}
}

@article{Peres.A,
  title = {Separability Criterion for Density Matrices},
  author = {Peres, Asher},
  journal = {Phys. Rev. Lett.},
  volume = {77},
  issue = {8},
  pages = {1413--1415},
  numpages = {0},
  year = {1996},
  month = {Aug},
  publisher = {American Physical Society},
  doi = {10.1103/PhysRevLett.77.1413},
  url = {https://link.aps.org/doi/10.1103/PhysRevLett.77.1413}
}

@article{Horodecki.MPR,
title = {Separability of mixed states: necessary and sufficient conditions},
journal = {Physics Letters A},
volume = {223},
number = {1},
pages = {1-8},
year = {1996},
issn = {0375-9601},
doi = {https://doi.org/10.1016/S0375-9601(96)00706-2},
url = {https://www.sciencedirect.com/science/article/pii/S0375960196007062},
author = {Michał Horodecki and Paweł Horodecki and Ryszard Horodecki}
}

@article{construction,
title = {Construction of {PPT} entangled state and its detection by using second-order moment of the partial transposition},
journal = {Physics Letters A},
volume = {567},
pages = {131195},
year = {2026},
issn = {0375-9601},
doi = {https://doi.org/10.1016/j.physleta.2025.131195},
url = {https://www.sciencedirect.com/science/article/pii/S0375960125009752},
author = {Rohit Kumar and Satyabrata Adhikari}
}

@article{Horodecki.P,
title = {Separability criterion and inseparable mixed states with positive partial transposition},
journal = {Physics Letters A},
volume = {232},
number = {5},
pages = {333-339},
year = {1997},
issn = {0375-9601},
doi = {https://doi.org/10.1016/S0375-9601(97)00416-7},
url = {https://www.sciencedirect.com/science/article/pii/S0375960197004167},
author = {Pawel Horodecki}
}

@article{multiqubit,
  title = {Bound entanglement in symmetric random induced states},
  author = {Louvet, J. and Damanet, F. and Bastin, T.},
  journal = {Phys. Rev. A},
  volume = {113},
  issue = {4},
  pages = {042440},
  numpages = {9},
  year = {2026},
  month = {Apr},
  publisher = {American Physical Society},
  doi = {10.1103/klk3-wjvn},
  url = {https://link.aps.org/doi/10.1103/klk3-wjvn}
}

@unpublished{multiqutrit,
      title={Entanglement in the {Dicke} subspace}, 
      author={Aabhas Gulati and Ion Nechita and Clément Pellegrini},
      eprint={2602.15800},
      archivePrefix={arXiv},
      primaryClass={quant-ph},
      url={https://arxiv.org/abs/2602.15800},
      note={Preprint: arXiv:2602.15800, 2026}
}

@article{powerful,
  title = {Bound entangled singlet-like states for quantum metrology},
  author = {P\'al, K\'aroly F. and T\'oth, G\'eza and Bene, Erika and V\'ertesi, Tam\'as},
  journal = {Phys. Rev. Res.},
  volume = {3},
  issue = {2},
  pages = {023101},
  numpages = {18},
  year = {2021},
  month = {May},
  publisher = {American Physical Society},
  doi = {10.1103/PhysRevResearch.3.023101},
  url = {https://link.aps.org/doi/10.1103/PhysRevResearch.3.023101}
}

@article{low_rank,
  title = {Separability and entanglement in ${C}^{2}\otimes{C}^{2}\otimes{C}^{N}$ composite quantum systems},
  author = {Karnas, Sini\ifmmode \check{s}\else \v{s}\fi{}a and Lewenstein, Maciej},
  journal = {Phys. Rev. A},
  volume = {64},
  issue = {4},
  pages = {042313},
  numpages = {11},
  year = {2001},
  month = {Sep},
  publisher = {American Physical Society},
  doi = {10.1103/PhysRevA.64.042313},
  url = {https://link.aps.org/doi/10.1103/PhysRevA.64.042313}
}

@article{biseparable,
  title = {Separability in $2\times{N}$ composite quantum systems},
  author = {Kraus, B. and Cirac, J. I. and Karnas, S. and Lewenstein, M.},
  journal = {Phys. Rev. A},
  volume = {61},
  issue = {6},
  pages = {062302},
  numpages = {10},
  year = {2000},
  month = {May},
  publisher = {American Physical Society},
  doi = {10.1103/PhysRevA.61.062302},
  url = {https://link.aps.org/doi/10.1103/PhysRevA.61.062302}
}

@article{444,
title = {Searching for extremal {PPT} entangled states},
journal = {Optics Communications},
volume = {283},
number = {5},
pages = {805-813},
year = {2010},
note = {Quo vadis Quantum Optics?},
issn = {0030-4018},
doi = {https://doi.org/10.1016/j.optcom.2009.10.050},
url = {https://www.sciencedirect.com/science/article/pii/S0030401809010426},
author = {Remigiusz Augusiak and Janusz Grabowski and Marek Kuś and Maciej Lewenstein}
}

@article{Chen_2013,
doi = {10.1088/1751-8113/46/27/275304},
url = {https://doi.org/10.1088/1751-8113/46/27/275304},
year = {2013},
month = {jun},
publisher = {IOP Publishing},
volume = {46},
number = {27},
pages = {275304},
author = {Chen, Lin and Đoković, Dragomir Ž},
title = {Separability problem for multipartite states of rank at most 4},
journal = {Journal of Physics A: Mathematical and Theoretical}
}

@article{prod,
  author  = {Kiem, Young-Hoon and Kye, Seung-Hyeok and Na, Joohan},
  title   = {Product Vectors in the Ranges of Multi-Partite States with Positive Partial Transposes and Permanents of Matrices},
  journal = {Communications in Mathematical Physics},
  year    = {2015},
  volume  = {338},
  number  = {2},
  pages   = {621--639},
  month   = {Sep},
  doi     = {10.1007/s00220-015-2385-x},
  url     = {https://doi.org/10.1007/s00220-015-2385-x}
}

@article{Bennett,
  title = {Unextendible Product Bases and Bound Entanglement},
  author = {Bennett, Charles H. and DiVincenzo, David P. and Mor, Tal and Shor, Peter W. and Smolin, John A. and Terhal, Barbara M.},
  journal = {Phys. Rev. Lett.},
  volume = {82},
  issue = {26},
  pages = {5385--5388},
  numpages = {0},
  year = {1999},
  month = {Jun},
  publisher = {American Physical Society},
  doi = {10.1103/PhysRevLett.82.5385},
  url = {https://link.aps.org/doi/10.1103/PhysRevLett.82.5385}
}

@article{3x3upb,
  author    = {DiVincenzo, David P. and Mor, Tal and Shor, Peter W. and Smolin, John A. and Terhal, Barbara M.},
  title     = {Unextendible Product Bases, Uncompletable Product Bases and Bound Entanglement},
  journal   = {Communications in Mathematical Physics},
  volume    = {238},
  number    = {3},
  pages     = {379--410},
  year      = {2003},
  month     = jul,
  doi       = {10.1007/s00220-003-0877-6},
  url       = {https://doi.org/10.1007/s00220-003-0877-6},
  issn      = {1432-0916}
}

@article{3x3classify,
  title = {Low-rank extremal positive-partial-transpose states and unextendible product bases},
  author = {Leinaas, Jon Magne and Myrheim, Jan and Sollid, Per \O{}yvind},
  journal = {Phys. Rev. A},
  volume = {81},
  issue = {6},
  pages = {062330},
  numpages = {6},
  year = {2010},
  month = {Jun},
  publisher = {American Physical Society},
  doi = {10.1103/PhysRevA.81.062330},
  url = {https://link.aps.org/doi/10.1103/PhysRevA.81.062330}
}

@article{Lin3x3,
    author = {Chen, Lin and Đoković, Dragomir Ž},
    title = {Description of rank four entangled states of two qutrits having positive partial transpose},
    journal = {Journal of Mathematical Physics},
    volume = {52},
    number = {12},
    pages = {122203},
    year = {2011},
    month = {12},
    issn = {0022-2488},
    doi = {10.1063/1.3663837},
    url = {https://doi.org/10.1063/1.3663837},
    eprint = {https://pubs.aip.org/aip/jmp/article-pdf/doi/10.1063/1.3663837/15777880/122203_1_online.pdf},
}

@article{concurrence,
  title = {Concurrence of multiqubit bound entangled states constructed from unextendible product bases},
  author = {Nawareg, Mohamed},
  journal = {Phys. Rev. A},
  volume = {101},
  issue = {3},
  pages = {032342},
  numpages = {10},
  year = {2020},
  month = {Mar},
  publisher = {American Physical Society},
  doi = {10.1103/PhysRevA.101.032342},
  url = {https://link.aps.org/doi/10.1103/PhysRevA.101.032342}
}

@article{max,
title = {Unextendible product bases, bound entangled states, and the range criterion},
journal = {Physics Letters A},
volume = {386},
pages = {126992},
year = {2021},
issn = {0375-9601},
doi = {https://doi.org/10.1016/j.physleta.2020.126992},
url = {https://www.sciencedirect.com/science/article/pii/S0375960120308598},
author = {Pratapaditya Bej and Saronath Halder}
}

@article{Johnston,
doi = {10.1088/1751-8113/47/42/424034},
url = {https://doi.org/10.1088/1751-8113/47/42/424034},
year = {2014},
month = {oct},
publisher = {IOP Publishing},
volume = {47},
number = {42},
pages = {424034},
author = {Johnston, Nathaniel},
title = {The structure of qubit unextendible product bases},
journal = {Journal of Physics A: Mathematical and Theoretical}
}

@article{general_position,
doi = {10.1088/1751-8113/48/4/045303},
url = {https://doi.org/10.1088/1751-8113/48/4/045303},
year = {2014},
month = {dec},
publisher = {IOP Publishing},
volume = {48},
number = {4},
pages = {045303},
author = {Ha, Kil-Chan and Kye, Seung-Hyeok},
title = {Multi-partite separable states with unique decompositions and construction of three qubit entanglement with positive partial transpose},
journal = {Journal of Physics A: Mathematical and Theoretical}
}

@article{WandGHZ,
  title = {Three qubits can be entangled in two inequivalent ways},
  author = {D\"ur, W. and Vidal, G. and Cirac, J. I.},
  journal = {Phys. Rev. A},
  volume = {62},
  issue = {6},
  pages = {062314},
  numpages = {12},
  year = {2000},
  month = {Nov},
  publisher = {American Physical Society},
  doi = {10.1103/PhysRevA.62.062314},
  url = {https://link.aps.org/doi/10.1103/PhysRevA.62.062314}
}

@article{SLOCC,
doi = {10.1088/1751-8121/aba576},
url = {https://doi.org/10.1088/1751-8121/aba576},
year = {2020},
month = {aug},
publisher = {IOP Publishing},
volume = {53},
number = {38},
pages = {385302},
author = {Li, Dafa and Guo, Yu},
title = {Local unitary equivalence of the {SLOCC} class of three qubits},
journal = {Journal of Physics A: Mathematical and Theoretical}
}

@article{GHZ,
  title = {Locality of three-qubit {Greenberger-Horne-Zeilinger-symmetric} states},
  author = {Zhu, Dian and He, Gang-Gang and Zhang, Fu-Lin},
  journal = {Phys. Rev. A},
  volume = {105},
  issue = {6},
  pages = {062202},
  numpages = {7},
  year = {2022},
  month = {Jun},
  publisher = {American Physical Society},
  doi = {10.1103/PhysRevA.105.062202},
  url = {https://link.aps.org/doi/10.1103/PhysRevA.105.062202}
}

@article{Lorentz_invariant,
  title = {Extremal states of positive partial transpose in a system of three qubits},
  author = {Steensgaard Garberg, \O{}yvind and Irgens, B\o{}rge and Myrheim, Jan},
  journal = {Phys. Rev. A},
  volume = {87},
  issue = {3},
  pages = {032302},
  numpages = {15},
  year = {2013},
  month = {Mar},
  publisher = {American Physical Society},
  doi = {10.1103/PhysRevA.87.032302},
  url = {https://link.aps.org/doi/10.1103/PhysRevA.87.032302}
}

@article{Bell_states,
  title = {Optimal creation of entanglement using a two-qubit gate},
  author = {Kraus, B. and Cirac, J. I.},
  journal = {Phys. Rev. A},
  volume = {63},
  issue = {6},
  pages = {062309},
  numpages = {8},
  year = {2001},
  month = {May},
  publisher = {American Physical Society},
  doi = {10.1103/PhysRevA.63.062309},
  url = {https://link.aps.org/doi/10.1103/PhysRevA.63.062309}
}

@article{pure,
  author = {Usha Devi, A. R. and Sudha and Shenoy, H. Akshata and Karthik, H. S. and Karthik, B. N.},
  title = {Lorentz invariants of pure three-qubit states},
  journal = {Quantum Information Processing},
  year = {2024},
  volume = {23},
  number = {7},
  pages = {264},
  doi = {10.1007/s11128-024-04454-2},
  url = {https://doi.org/10.1007/s11128-024-04454-2},
  issn = {1573-1332}
}

@article{three_qubit_UPB,
  author    = {Bravyi, S. B.},
  title     = {Unextendible Product Bases and Locally Unconvertible Bound Entangled States},
  journal   = {Quantum Information Processing},
  year      = {2004},
  volume    = {3},
  number    = {6},
  pages     = {309--329},
  month     = dec,
  doi       = {10.1007/s11128-004-7076-z},
  issn      = {1573-1332},
  url       = {https://doi.org/10.1007/s11128-004-7076-z}
}

@article{experiment1,
  title = {Experimental generation of pseudo-bound-entanglement},
  author = {Kampermann, Hermann and Bru\ss{}, Dagmar and Peng, Xinhua and Suter, Dieter},
  journal = {Phys. Rev. A},
  volume = {81},
  issue = {4},
  pages = {040304(R)},
  numpages = {4},
  year = {2010},
  month = {Apr},
  publisher = {American Physical Society},
  doi = {10.1103/PhysRevA.81.040304},
  url = {https://link.aps.org/doi/10.1103/PhysRevA.81.040304}
}

@article{Toth2007,
  title = {Optimal Spin Squeezing Inequalities Detect Bound Entanglement in Spin Models},
  author = {T\'oth, G\'eza and Knapp, Christian and G\"uhne, Otfried and Briegel, Hans J.},
  journal = {Phys. Rev. Lett.},
  volume = {99},
  issue = {25},
  pages = {250405},
  numpages = {4},
  year = {2007},
  month = {Dec},
  publisher = {American Physical Society},
  doi = {10.1103/PhysRevLett.99.250405},
  url = {https://link.aps.org/doi/10.1103/PhysRevLett.99.250405}
}
\end{document}